\documentclass[lettersize,onecolumn]{IEEEtran}
\usepackage{amsfonts,amsthm,amssymb, amsmath}
\usepackage{graphicx} 
\usepackage[margin=1in]{geometry}
\usepackage{xcolor}
\usepackage[hidelinks]{hyperref}
\usepackage{authblk}
\usepackage{stfloats}
\usepackage{cite}

\newtheorem{theorem}{Theorem}[section]
\newtheorem{lemma}[theorem]{Lemma}
\newtheorem{proposition}[theorem]{Proposition}
\newtheorem{corollary}[theorem]{Corollary}

\theoremstyle{definition}
\newtheorem{definition}[theorem]{Definition} 
\newtheorem{remark}[theorem]{Remark}
\newtheorem{example}[theorem]{Example}

\newcommand{\F}{{\mathbb{F}}}

\definecolor{ccamps}{rgb}{0.44, 0.0, 1.0}

\definecolor{crod}{RGB}{200, 120, 0}

\title{Transversal gates for quantum CSS codes}

\author[1,2]{Eduardo Camps-Moreno}
\author[2]{Hiram H. L\'opez}
\author[2]{Gretchen L. Matthews}
\author[3]{Narayanan Rengaswamy}
\author[2]{Rodrigo San-Jos\'e}

\affil[1]{Université de Bordeaux, Talence, France \footnote{Email: eduardo.camps-moreno@math.u-bordeaux.fr}}
\affil[2]{Department of Mathematics\\ Virginia Tech\\ Blacksburg, VA USA\footnote{Emails: \{gmatthews, hhlopez, rsanjose\}@vt.edu}}
\affil[3]{Department of Electrical and Computer Engineering\\ University of Arizona\\ Tucson, AZ USA \footnote{Email: narayananr@arizona.edu}}

\begin{document}

\maketitle

\begin{abstract}
In this paper, we focus on the problem of computing the set of diagonal transversal gates fixing a CSS code. We determine the logical actions of these gates and characterize the subgroups that induce tensor products of single-qubit logical gates and logical identities. We describe these groups through explicit linear equations given by componentwise products of basis vectors of the underlying classical codes, which is the key difference from previous approaches. We compute the complete set of diagonal transversal gates fixing the code, and of transversal logical gates, for any CSS code arising from monomial codes, a family that includes decreasing monomial codes and polar codes. As a consequence, we recover and extend some results in the literature on CSS-T codes, triorthogonal codes, and divisible codes. We also extend these results to diagonal gates acting transversally on several code blocks and use them to characterize addressable gates.
\end{abstract}

\section{Introduction}

Over the past few decades, quantum computing has attracted increasing attention, mainly due to the existence of quantum algorithms that achieve exponential speed-ups over the best-known classical algorithms, e.g., Shor's algorithm for integer factorization \cite{shorPrimeFactorization}. One of the main challenges for reliable large-scale quantum computing is to develop techniques to mitigate and correct errors arising from the effect of noise and decoherence on qubits. Quantum error-correcting codes (QECCs), first introduced by Shor \cite{shorFirstQECC}, allow using a number, $n$, of physical qubits to encode $k<n$ logical qubits, in such a way that the logical qubits are resistant to a certain number of errors. The CSS construction, named after Calderbank, Shor, and Steane, enables the construction of QECCs from classical codes \cite{calderbankgoodquantum,cssoriginal2}.

For fault-tolerant quantum computing, besides knowing how to correct errors on the logical qubits, we must also identify how to perform operations (logical gates) on the logical qubits. The Clifford group, together with a non-Clifford gate such as the $T$ gate, forms a universal gate set, that is, any other gate can be arbitrarily approximated using gates from this set. Due to the Eastin-Knill theorem~\cite{eastinknill}, it is not possible to implement a universal gate set transversally with a single QECC. To circumvent this limitation, a standard approach has been to consider a QECC implementing the Clifford group transversally, and then to use magic state distillation~\cite{Bravi-12} and state injection to perform the $T$ gate. 

Since many magic state distillation protocols require codes implementing the $T$ gate transversally, recent work has focused on finding classical codes such that the corresponding CSS code implements the $T$ gate transversally, e.g., triorthogonal codes~\cite{Bravi-12, haah_classification_triorthogonal, camps_triorthogonal, kim_triorthogonal_selfdual, albert_transversal_gates,sanjose_good_triorthogonal}. Relaxing the conditions, one can also study CSS-T quantum codes, which are CSS codes $Q$ such that $T^{\otimes n}Q=Q$, and the related CSS-T pairs of classical codes \cite{Rengaswamy-20, campsCSST, berardini_asymptotically_good_csst, felice_csst}. These types of conditions have also been studied for higher-dimensional systems, using $q$-ary classical codes \cite{golowich_asymptotically_good_quantum_transversal, prakash_qutrit_magic_state, wills_constant_overhead_msd}.

In this paper, we focus on the problem of computing the set of diagonal transversal gates fixing a CSS code. We determine the logical actions of these gates and characterize the subgroups that induce tensor products of single-qubit logical gates and logical identities. We describe explicitly the set of equations defining these groups, in terms of componentwise products of basis vectors of the classical codes. This is the key difference from~\cite{Webster-23}, where these groups are computed algorithmically, as kernels over $\mathbb{Z}_N$ of matrices derived from a basis of the code. This structural information allows us to recover some results from \cite{Rengaswamy-20, campsCSST} for CSS-T codes, and from \cite{Bravi-12, Haah-18} for triorthogonal codes and divisible codes.

A primary advantage of the approach we present here is that, for well-structured codes, we can compute the whole group of diagonal transversal gates fixing the code. We compute the transversal gates for any CSS code arising from monomial codes, which includes decreasing monomial codes. Polar codes,  which have been considered in~\cite{krishna2018polar}, are decreasing monomial codes. We also recover some recent results for Reed-Muller codes~\cite{Barg}. 

We also extend these results to diagonal gates acting transversally across several CSS code blocks, including controlled-$S$ and $CCZ$ gates. We give explicit conditions for preserving the tensor product of the codes and determine the induced logical action. These results characterize whether a prescribed logical gate admits an addressable implementation within this family, by solving a system of linear equations over $\mathbb{Z}_N$. 

The paper is organized as follows. This section concludes with the notation to be used throughout the paper. Section \ref{sec:arbitrary_characters} describes in general the transversal operators acting on a CSS code and the relation with the sign of its stabilizers. In Section~\ref{sec:diagonal}, we describe how to compute the group of diagonal transversal gates of a CSS code and their logical action, and we determine which of them act as transversal logical gates and which act as logical identities. We also compute these groups explicitly for CSS codes arising from monomial codes, such as Reed--Muller and polar codes. In Section~\ref{sec:multiblock}, we extend these results to diagonal gates acting transversally on several code blocks, such as transversal $CCZ$ and controlled-$S$ gates, and we discuss addressable gates. A summary and open problems are given in Section \ref{sec:conclusion}. 

\medskip

\paragraph*{Notation} For $n\in\mathbb{N}$, we write $[n]$ for the set $\{1,\ldots,n\}$. We denote by $\mathbb{Z}_N$ the integers modulo $N$. For $N=2$, we identify $\mathbb{F}_2 \leftrightarrow \mathbb{Z}_2$ with the integers $\{0,1\}$. To differentiate the additive structures, we denote by $(\mathbb{Z}_2,+)$ the addition modulo 2 and by $(\mathbb{Z},\oplus)$ the standard integer addition. In other words,
\begin{align*}
1+1&=0 \text{ (addition modulo 2)},\\
1\oplus 1&=2 \text{ (integer addition)}.
\end{align*}
We also use $-$ and $\ominus$, respectively. More formally, given $N\in\mathbb{Z}$, we map $\mathbb{Z}_N$ to $\{0, 1, \ldots, N-1\}$ as sets, noting that the map is not a homomorphism. 

If $R$ is a field or a ring, $(v_1,\ldots,v_n)$ represents an element in $R^n$. Matrices of size $n\times m$ with entries in $R$ are denoted by $R^{n\times m}$. We denote as $\mathrm{diag}(a_1,\ldots,a_n)$ the diagonal $n\times n$ matrix $D$ with entries $D_{ii}=a_i$ and $D_{ij}= 0$ for any $i\neq j$.

We work with codes over $\mathbb{Z}_N$. A code $C$ is an additive subgroup of $\mathbb{Z}_N^n$, or equivalently, a submodule of $\mathbb{Z}_N^n$ over $\mathbb{Z}_N$, i.e., $C$ is closed under addition and product by elements of $\mathbb{Z}_N$. An element $c\in C$ is a codeword and its weight is $wt(c)=|\{i\in [n] :c_i\neq 0\}|$. The minimum distance of the code is $d(C)=\min\{wt(c) :c\in C\setminus\{0\}\}$. The dual of a code $C$ is 
$$C^{\perp_N}=\left\{v\in\mathbb{Z}_N^n :\sum_{i=1}^n c_iv_i=0,\;\forall c\in C\right\}.$$
If $N=2$, we omit the subindex and write $C^\perp$. The length of $C$ (as a module), denoted $\lambda(C)$, is the maximum length $r$ of a chain of submodules
$$C_0=\{0\}\subsetneq C_1\subsetneq C_2\subsetneq\cdots\subsetneq C_r=C.$$
The length of $\mathbb{Z}_{2^\ell}$ is $\ell$, and for any code $C\subseteq \mathbb{Z}_{2^\ell}^n$, we have $\lambda(C)+\lambda(C^{\perp_{2^\ell}})=n\ell$ (cf. \cite{coderings}). In fact, $|C|=2^{\lambda(C)}$, so $\lambda(C)$ plays the role of the dimension of a linear code.

In the context of quantum codes, we work with finite-dimensional complex vector spaces. We identify an orthonormal basis of $\mathbb{C}^2$ with $\{|0\rangle,|1\rangle\}$ and a corresponding orthonormal basis of $\mathbb{C}^{2^n}$ with $\{|u\rangle: u\in\mathbb{Z}_2^n\}$, where $|u\rangle = |u_1\rangle\otimes\cdots\otimes|u_n\rangle$. We also consider $(\mathbb{C}^2)^\ast$ and its corresponding orthonormal basis $\{\langle 0|,\langle 1|\}$, where $\langle u|(|v\rangle)=1$ if $u=v$ and $\langle u|(|v\rangle)=0$ otherwise. To simplify the notation, we write $\langle u| v\rangle$ instead, and we identify $\mathbb{C}^2$ with vectors as columns and $(\mathbb{C}^2)^\ast$ with vectors as rows. Any transformation $E\in\mathbb{C}^{2\times 2}$ can be written then as a linear combination of transformations $|v\rangle\langle u|$, which maps the vector $|u\rangle$ to $|v\rangle$. After fixing an orthonormal basis, all matrices can be written in terms of the chosen basis. In particular, we can write the Pauli operators as $2\times 2$ matrices:
$$
X=\begin{pmatrix}
0&1\\
1&0
\end{pmatrix}, \; Z=\begin{pmatrix}
1&0\\
0&-1
\end{pmatrix},\; Y=\begin{pmatrix} 0&-i\\ i&0\end{pmatrix},\; I=\begin{pmatrix} 1&0\\ 0&1\end{pmatrix}.$$
We extend this to $\mathbb{C}^{2^n}=\left(\mathbb{C}^2\right)^{\otimes n}$ by taking $u\in\mathbb{Z}_2^n$ and defining
$$
X(u)=X^{u_1}\otimes X^{u_2}\otimes \cdots \otimes X^{u_n}, \; Z(u)=Z^{u_1}\otimes Z^{u_2}\otimes \cdots \otimes Z^{u_n}.
$$

We call $\mathcal{P}_n$ the group generated by $\{\pm X(u)Z(v) : u,v\in\mathbb{Z}_2^n\}$. A quantum code is a vector subspace $Q\subseteq\mathbb{C}^{2^n}$. We say that an operator $E\in\mathbb{C}^{2^n\times 2^n}$ is detectable by $Q$ if there is a constant $\lambda_E$ such that for any $|\psi\rangle,|\varphi\rangle\in Q$, we have
$$\langle\psi|E|\varphi\rangle=\lambda_E\langle\psi|\varphi\rangle.$$

We say that $Q$ has minimum distance $d$ if for any element $E\in\mathcal{P}_n$ such that $E=E_1\otimes\cdots\otimes E_n$ and $|\{i\in [n]: E_i\neq I\}|\leq d-1$, $E$ is detectable by $Q$. If $Q$ has dimension $2^k$ and minimum distance $d$, we say that $Q$ is an $[[n,k,d]]_2$ code. In this paper, we only consider quantum codes whose dimension is a power of $2$.

Given a quantum code $Q\subseteq\mathbb{C}^{2^n}$ of dimension $2^k$,  an encoding map $\mathcal{E}:\mathbb{C}^{2^k}\rightarrow\mathbb{C}^{2^n}$, and  $u\in\mathbb{Z}_2^k$, we denote  $|u\rangle_L=\mathcal{E}(|u\rangle)$, and we refer to it as the logical state $u$. Given a unitary $U:\mathbb{C}^{2^n}\rightarrow\mathbb{C}^{2^n}$ such that $U(Q)=Q$, the description of $\mathcal{E}^{-1}\circ U\circ \mathcal{E}$ in terms of our orthonormal basis is called the logical action of $U$. 

\section{Transversal operators for CSS codes with arbitrary characters} \label{sec:arbitrary_characters}

We start by recalling the definition and the parameters of CSS codes. These can be found, for example, in \cite{Nielsen_Chuang_2010,ketkar_nonbinary_stabilizer}.

\begin{definition}\rm
    Let $C_2\subseteq C_1\subseteq\mathbb{F}_2^n$ and $y_x,y_z\in\mathbb{F}_2^n$. The CSS code $Q=Q(C_1,C_2,y_x,y_z)$ is the subspace of $(\mathbb{C}^2)^{\otimes n}$ formed by the vectors fixed by every element of the stabilizer group $\{\epsilon(u,v)X(u)Z(v): u\in C_2,\ v\in C_1^\perp\}$, where $\epsilon(u,v)=(-1)^{y_x\cdot u+y_z\cdot v}$.
    When unspecified, we assume $y_x=y_z=0$ and write $Q(C_1,C_2):=Q(C_1,C_2,0,0)$.
\end{definition}

The following well-known result provides the basic parameters of CSS codes.

\begin{lemma}\label{lem.propcss}
    Let $C_2\subseteq C_1\subseteq\mathbb{F}_2^n$, $y_x,y_z\in\mathbb{F}_2^n$, and $L\subseteq\mathbb{F}_2^n$ be a set of representatives of $C_1/C_2$. Then, the set $\{|C_2+w+y_z\rangle : w\in L\}$ is an orthonormal basis for $Q=Q(C_1,C_2,y_x,y_z)$, where 
    
$$|C_2+w+y_z\rangle=\frac{1}{\sqrt{|C_2|}}\sum_{v\in C_2}(-1)^{y_x\cdot v}|v+w+y_z\rangle.$$
These states also depend on $y_x$, which we omit from the notation. 
    Moreover, if $d_x=\min\{wt(w) : w\in C_1\setminus C_2\}$, $d_z=\min\{wt(w) : w\in C_2^\perp\setminus C_1^\perp\}$, $k_1=\dim C_1$, and $k_2=\dim C_2$, then $Q$ has parameters $[[n,k_1-k_2,\min\{d_x,d_z\}]]_2$.

\end{lemma}

\begin{definition}\rm
     Let $C_2\subseteq C_1\subseteq\mathbb{F}_2^n$,  $y_x,y_z\in\mathbb{F}_2^n$, $L\subseteq\mathbb{F}_2^n$ be a set of representatives of $C_1/C_2$, and $K=\dim C_1-\dim C_2$. The orthonormal basis $\{|C_2+w+y_z\rangle : w\in L\}$ consists of the corresponding logical states.  If $\beta=\{w_1,\ldots,w_{K}\}$ is a set of representatives whose classes span $C_1/C_2$, then for $v\in\mathbb{F}_2^K$, we write $$|v\rangle_L=\left|C_2+y_z+\sum_{i=1}^K v_iw_i\right\rangle.$$
     We say that $\{w_1,\ldots,w_K\}$ is an encoding of $Q$.
\end{definition}

We are interested in diagonal transversal gates $U$ fixing the code $Q$. Up to a global phase, we may assume that $U$ is of the form $\displaystyle U=\bigotimes_{i=1}^n \mathrm{diag}(1,\omega_i)$, where each $\omega_i$ is a complex number of modulus one, and such that $UQ=Q$. The following result shows that, for codes of distance at least two, the phases appearing on the diagonal can be taken to be roots of unity of order equal to a power of $2$, up to a gate acting trivially on $Q$ and a global phase.

\begin{lemma}\cite[Proposition 3]{Anderson-14}
If $\displaystyle U=\bigotimes_{i=1}^n \mathrm{diag}(1,\omega_i)$ fixes a binary CSS code $Q$ of distance at least two, then, up to a global phase, $U$ acts on $Q$ as $\bigotimes_{i=1}^n \text{diag}(1,\omega'_i)$, where each $\omega_i'$ is a $2^\ell$-th root of unity for some nonnegative $\ell$.
\end{lemma}

In light of this, we focus on the case where the roots have powers of 2. For a given $\ell\in\mathbb{N}$, we take $N=2^\ell$ and $\omega=e^{2\pi \mathrm{i}/N}$. For $a\in\mathbb{Z}_N$ and $b\in\mathbb{Z}_N^n$, we denote $U(a)=\mathrm{diag}(1,\omega^a)$ and ${U(b)=\displaystyle \bigotimes_{i=1}^n U(b_i)}$. Diagonal transversal operators can act as logical operators, in which case their action is determined by their action on the logical states $|v\rangle_L$ of the code.

\begin{lemma}\label{Lemma.action}
Take $\omega=e^{2\pi\mathrm{i}/N}$ for some $N=2^\ell$, and let $b\in\mathbb{Z}_N^n$ be such that $U(b)Q=Q$, for a CSS code $Q=Q(C_1,C_2,y_x,y_z)$. Fix an encoding $\{w_1,\ldots,w_K\}\subseteq C_1$ of $Q$ and define $B$ as the matrix whose $i$-th column is $w_i$. Then, $U(b)$ acts in the logical system as
$$\overline{U(b)}=\sum_{v\in\mathbb{F}_2^K}\omega^{b\cdot (y_z+Bv)}|v\rangle\langle v|.$$
\end{lemma}

\begin{proof}
    Let $v\in\mathbb{F}_2^K$ and $w=\sum_{i=1}^K v_iw_i$. Then
    $$U(b)|v\rangle_L=\frac{1}{\sqrt{|C_2|}}\sum_{u\in C_2} (-1)^{y_x\cdot u } U(b)|u+w+y_z\rangle=\frac{1}{\sqrt{|C_2|}}\sum_{u\in C_2}(-1)^{y_x\cdot u }\omega^{b\cdot (y_z+u+w)}|u+w+y_z\rangle.$$
    Since $U(b)Q=Q$, for each $t\in\mathbb{F}_2^K$, there is a constant $\alpha_t$ such that    $$U(b)|v\rangle_L=\sum_{t\in\mathbb{F}_2^K}\alpha_t|t\rangle_L.$$
Since $|t\rangle_L$ is supported on the coset $C_2+y_z+\sum_{i=1}^Kt_iw_i$, and these cosets are pairwise disjoint, the orthonormality of the set $\{|a\rangle : a\in\mathbb{F}_2^n\}$ gives $\alpha_t=0$ for all $t\neq v$ and, comparing the coefficients of $|u+w+y_z\rangle$,
    \begin{equation}\label{eq:coeffs}
    \omega^{b\cdot(y_z+u+w)}=\alpha_v\ \qquad \text{ for all } u\in C_2.
    \end{equation}
    In particular, for $u=0$, we get $\alpha_v=\omega^{b\cdot(y_z+w)}$, from which the conclusion follows.
\end{proof}

    Note that if $u\in C_1$, then lifting $y_z+u$ from $\mathbb{F}_2$ to $\mathbb{Z}_N$ does not allow us to express $b\cdot (y_z+u)$ as a linear combination of $b\cdot y_z$ and $b\cdot w$ for $w$ in a basis of $C_1$. For instance, for $w_1,w_2\in\mathbb{F}_2^n$, we have $b\cdot(w_1+w_2)=b\cdot w_1+b\cdot w_2-2\,b\cdot(w_1\star w_2)$, where $\star$ is the componentwise product, and the last term is not determined by $b\cdot w_1$ and $b\cdot w_2$. Hence, knowing the phase values on basis vectors does not generally determine the phase values on their binary sums: higher-order componentwise products also enter. These products can themselves be computed from the basis. In the next section, we see that this is the main challenge to compute the set of diagonal operators fixing $Q$. The actions that can be understood just by the action on a choice of a basis are precisely the transversal logical operators that arise from some $U(b)$.
    \begin{definition}\rm
            We say that a logical operator $\overline{E}$ of $Q=Q(C_1,C_2,y_x,y_z)$ is transversal if for a fixed encoding of $Q$, we have $\displaystyle \overline{E}=\bigotimes_{i=1}^K E_i$ for some unitary operators $E_i\in\mathbb{C}^{2\times 2}$. 
    \end{definition}

While the diagonal transversal operators fixing a CSS code do depend on the specific $y_z$, there is always a bijection between those fixing $Q(C_1,C_2,y_x,y_z)$ and those fixing $Q(C_1,C_2,0,0)$.

\begin{proposition}\label{Prop.yz0}
    Let $Q=Q(C_1,C_2,y_x,y_z)$ be a CSS code. For $b\in\mathbb{Z}_N^n$, let $b_{y_z}$ be the vector given by
    $$(b_{y_z})_i=\begin{cases}
b_i,  & \text{if } (y_z)_i = 0,\\
-b_i, & \text{if } (y_z)_i = 1.
\end{cases}$$
Then, $U(b)$ fixes $Q$ if and only if $U(b_{y_z})$ fixes $Q'=Q(C_1,C_2)$.
\end{proposition}

\begin{proof}
Since $Z(y_x)$ commutes with $U(b)$ and $Z(y_x)Q(C_1,C_2,0,y_z)=Q$, we may assume $y_x=0$. 
    Take $a\in\mathbb{Z}_N$. We have $(X\, U(a)\, X)=\omega^aU(-a)$. Thus, $X(y_z)U(b)X(y_z)=\omega^{y_z\cdot b}U(b_{y_z})$. We can observe that $X(y_z)$ sends any of the orthonormal bases of $Q$ described in Lemma~\ref{lem.propcss} to an orthonormal basis of the same type for $Q'$. Thus,
    $$Q'=X(y_z)Q=X(y_z)U(b)Q=X(y_z)U(b)X(y_z) Q'=\omega^{y_z\cdot b}U(b_{y_z}) Q'=U(b_{y_z})Q',$$
    and the converse follows in the same way, since $(b_{y_z})_{y_z}=b$. This completes the proof.
\end{proof}

Even more, the logical actions determined by these operators are the same after the bijection. 

\begin{corollary}\label{coro.factor}
    Let $U(b)$ be a transversal diagonal operator fixing the CSS code $Q=Q(C_1,C_2,y_x,y_z)$ and take $Q'=Q(C_1,C_2)$.
    \begin{enumerate}
    \item The logical action determined by $U(b)$ over $Q$ is the same as the one determined by $\omega^{y_z\cdot b}U(b_{y_z})$ over $Q'$.

    \item If $U(b)$ produces a transversal logical operator over $Q$, then $U(b_{y_z})$ produces a transversal logical one over $Q'$.

    \item If $U(b)$ is an identity over $Q$, then $U(b_{y_z})$ is an identity over $Q'$.
    \end{enumerate}
\end{corollary}

\begin{proof}
 For binary $x$, we have $b\cdot(y_z+x)=b\cdot y_z+b_{y_z}\cdot x$. Hence, by Lemma~\ref{Lemma.action}, the logical action of $U(b)$ over $Q$ is $\omega^{y_z\cdot b}\sum_{v}\omega^{b_{y_z}\cdot Bv}|v\rangle\langle v|$, which is the one of $\omega^{y_z\cdot b}U(b_{y_z})$ over $Q'$. Items 2 and 3 follow, since a global phase does not affect being a tensor product, and, if $U(b)$ is the identity on $Q$, then $\omega^{y_z\cdot b}=1$ (consider $|0\rangle_L$).
\end{proof}

This proves that to understand the diagonal transversal operators of $Q(C_1,C_2,y_x,y_z)$ for any $y_x,y_z\in\mathbb{F}_2^n$, we just need to understand the behavior of those of $Q(C_1,C_2)$.

\begin{remark}\rm
    Proposition~\ref{Prop.yz0} is, in the first sense, a general statement for any logical operator. In other words, with the notation established in Proposition~\ref{Prop.yz0}, if $E$ is an operator such that $EQ=Q$, then $E'=WEW^\dagger$, with $W=Z(y_x)X(y_z)$, is an operator such that $E'Q'=Q'$. Thus, in general, we need to focus on computing the logical operators of $Q(C_1,C_2)$ to understand those with any other $y_z\neq 0$. 
\end{remark}

\section{Diagonal transversal operators for CSS codes} \label{sec:diagonal}

From now on, we consider $y_z=0$. We also take $y_x = 0$ as it is irrelevant for our computations. Let $C_2\subseteq C_1\subseteq\mathbb{F}_2^n$ be two linear codes and $Q=Q(C_1,C_2)$ be the corresponding CSS code.
\begin{definition}\rm
    Fix an encoding of $Q$ and take $N=2^\ell$ for some positive integer $\ell$. The group of diagonal transversal operators fixing $Q$, the group of transversal logical operators of $Q$, and the group of logical identities that can be realized as diagonal transversal operators of $Q$ are defined, respectively, by
    \begin{align*}
    H_N(Q)&=\{b\in\mathbb{Z}_N^n :U(b)Q=Q\},\\
    T_N(Q)&=\{b\in H_N(Q) :U(b)\ \text{is a transversal logical operator}\}, \text{ and }\\
    Id_N(Q)&=\{b\in H_N(Q) :U(b)v=v,\ \forall v\in Q\}.
    \end{align*}
   Observe that
    $$Id_N(Q)\subseteq T_N(Q)\subseteq H_N(Q).$$
    We write $H_N, T_N$, and $Id_N$ when $Q$ is clear from the context.
\end{definition}

Clearly, the three sets are closed under the sum and multiplication by elements in the corresponding ring. We can observe that $T_N$ depends on the fixed encoding, while $H_N$ and $Id_N$ do not (even though the last is contained in $T_N$). However, fixing an encoding helps us to compute these three groups, as we show below.

For $u,v\in\mathbb{Z}_N^n$, $u\star v$ denotes the componentwise product of $u$ and $v$. For sets $A,B\subseteq\mathbb{Z}_N^n$, we write $A\star B=\{u\star v :u\in A,\ v\in B\}$, and for $r\geq 1$, we write
$$A^{(r)}=\{v_1\star\cdots\star v_r :v_i\in A, v_i\neq v_j,\ 1\leq i<j\leq r\}.$$
We define $A^{(0)}=\{(1,\ldots,1)\}$ and $A^{\perp_N}=\{w\in\mathbb{Z}_N^n :v\cdot w=0\ \text{ for all } v\in A\}$. For a binary linear code $C\subseteq\mathbb{F}_2^n$, we define its $r$-th Schur power by
\[
C^r=
\mathrm{Span}_{\mathbb{F}_2}
\{v_1\star\cdots\star v_r:v_1,\ldots,v_r\in C\},
\]
for $r\geq 1$, and set $C^0=\mathrm{Span}_{\mathbb{F}_2}\{(1,\ldots,1)\}$. Unlike $C^{(r)}$, this definition allows repeated factors and takes the linear span of their products.

\begin{theorem}\label{thm.main}
Let $C_2\subseteq C_1\subseteq\mathbb{Z}_2^n$ be linear codes with bases $\beta_2\subseteq\beta_1$, respectively, and $Q=Q(C_1,C_2)$ the corresponding CSS code. If $N=2^\ell$, then
$$H_N=\left(\beta_2\star\bigcup_{i=0}^{\ell-1}\left(2^i\beta_1^{(i)}\right)\right)^{\perp_N}.$$
\end{theorem}

    \begin{proof} Recall that we distinguish the sum in $\mathbb{Z}_N$ using $\oplus$ from the sum in $\mathbb{Z}_2$, denoted by $+$. Thus, for $v,w\in\mathbb{Z}_2^n$, we have
    \begin{equation}\label{eq.2tom}
    v+w=v\oplus w\ominus 2v\star w.
    \end{equation}
    We use Equation~(\ref{eq.2tom}) several times to expand an element of $C_1$ in terms of the selected basis. Take $b\in \mathbb{Z}_N^n$. By the proof of Lemma \ref{Lemma.action} (in particular, Equation \eqref{eq:coeffs}), we have that $b\in H_N$ if and only if for all $v\in C_2$ and $w\in C_1$,
$$\omega^{b\cdot(v+w)}=\omega^{b\cdot w}.$$
    Using Equation (\ref{eq.2tom}), the previous statement is valid if and only if for all $v\in C_2$ and any $w\in C_1$,
\begin{equation}\label{eq.in}
    b\cdot(v\ominus 2 w\star v)=0\bmod N.
    \end{equation} 
    Now we assume $b\in H_N$. In particular, Equation (\ref{eq.in}) implies that for $w=0$ and any $v\in C_2$,
    \begin{equation}\label{eq.i0}
    b\cdot v=0\bmod N.
    \end{equation}

    Combining Equations (\ref{eq.in}) and ($\ref{eq.i0}$), we get
    \begin{equation}\label{eq.i1}
    2b\cdot(v\star w)=0\bmod N.
    \end{equation}

    If $w_1,w_2\in C_1$ and $w=w_1+w_2$, plugging this $w$ into Equation~(\ref{eq.i1}) and expanding using Equation~(\ref{eq.2tom}), we get $2b\cdot(v\star w_1)+2b\cdot(v\star w_2)-4b\cdot(v\star w_1\star w_2)=0\bmod N$, so, by Equation~(\ref{eq.i1}), $4b\cdot(v\star w_1\star w_2)=0\bmod N$ for any $v\in C_2$ and $w_1,w_2\in C_1$.

    Repeating this argument (by induction on the number of factors), we obtain the following equations for any $v\in C_2$ and any $w_1,\ldots,w_{\ell-1}\in C_1$:

    $$4b\cdot(v\star w_1\star w_2)=0\bmod N$$
    $$8b\cdot(v\star w_1\star w_2\star w_3)=0\bmod N$$
    $$\vdots$$
$$2^{\ell-1}b\cdot(v\star w_1\star \cdots\star w_{\ell-1})=0\bmod N.$$
    We can continue beyond $\ell$. Since $N=2^\ell$, the equations are trivially satisfied. Therefore, for any $v\in \beta_2$ and $w\in 2^i\beta_1^{(i)}$, $b\cdot v\star w=0\bmod N$, which proves that $\mathrm{H}_N\subset \left( \beta_2\star\bigcup_{i=0}^{\ell-1}(2^i\beta_1^{(i)})\right)^{\perp_N}$.

To conclude, write $v=\sum_{u\in S_2}u$ and $w=\sum_{u\in S_1}u$, where $S_2\subseteq\beta_2$ and $S_1\subseteq\beta_1$. For every nonempty $A\subseteq\beta_1$, set $p_A=\bigstar_{u\in A}u$. Repeated application of Equation~\eqref{eq.2tom} gives
\[
\begin{aligned}
v\ominus 2w\star v=
\bigoplus_{\emptyset\neq A\subseteq S_2}
(-2)^{|A|-1}p_A\oplus
\bigoplus_{\substack{
\emptyset\neq A\subseteq S_2\\
\emptyset\neq B\subseteq S_1}}
(-2)^{|A|+|B|-1}(p_A\star p_B).
\end{aligned}
\]
    Since each one of the terms in this sum is a multiple of an element in $2^i\beta_2\star\beta_1^{(i)}$ for some $i\geq 0$, and it vanishes modulo $N$ if $i\geq\ell$, if $b\in\left( \beta_2\star\bigcup_{i=0}^{\ell-1}(2^i\beta_1^{(i)})\right)^{\perp_N}$, we have
    $$b\cdot(v\ominus 2w\star v)=0\bmod N.$$
    By Equation~(\ref{eq.in}), this implies $b\in\mathrm{H}_N$, concluding the proof.      
\end{proof}

\begin{example}\label{ex.T}\rm
Let $N=8$ and $b=(1,\ldots,1)$, so that $U(b)=T^{\otimes n}$ with $T=\mathrm{diag}(1,e^{\pi\mathrm{i}/4})$. Since $b\cdot x=wt(x)$ for every $x\in\mathbb{F}_2^n$, Theorem~\ref{thm.main} states that $T^{\otimes n}$ fixes $Q=Q(C_1,C_2)$ if and only if, for every $v\in\beta_2$ and all distinct $w,w'\in\beta_1$,
\begin{enumerate}
    \item[(i)] $wt(v)=0\bmod 8$,
    \item[(ii)] $wt(v\star w)=0\bmod 4$, and
    \item[(iii)] $wt(v\star w\star w')=0\bmod 2$.
\end{enumerate}
Since $C_1^2$ is spanned by $\beta_1\cup\beta_1^{(2)}$ and $\beta_2\subseteq\beta_1$, condition (iii) applied to the pair $\{v,w\}$ shows that $wt(v\star w)=wt(v\star v\star w)$ is even for every $w\in\beta_1\setminus\{v\}$. Thus, (iii), together with the fact that $C_2$ is even (which follows from (i)), is equivalent to $C_2\subseteq (C_1^2)^\perp$. By~\cite[Theorem~2.3]{campsCSST}, this inclusion characterizes the CSS-T pairs of~\cite{calderbankclassicalcsst,campsCSST}, i.e., the pairs such that $C_2$ is even and, for every $v\in C_2$, the shortening of $C_1^\perp$ to the support of $v$ contains a self-dual code. Therefore, up to the evenness of $C_2$, the CSS-T pair condition corresponds to the last level, $i=\ell-1=2$, of Theorem~\ref{thm.main}, while (i) and (ii), coming from the first two levels (levels $i=0$ and $i=1$), are the additional conditions needed for $T^{\otimes n}$ to fix $Q$, i.e., the conditions to obtain a CSS-T quantum code. Once $C_2\subseteq (C_1^2)^\perp$ has been verified, conditions (i) and (ii) only require $k_2(k_1+1)$ weight computations, where $k_j=\dim C_j$, whereas (iii) consists of $k_2\binom{k_1}{2}$ conditions.
\end{example}

\begin{remark}\label{rem.complexity}\rm
In the next results, we demonstrate that our approach provides more structural information about the set $H_N$ for a given code $C_1$ than that given in~\cite{Webster-23}. Similarly, we can obtain certain properties of $C_1$ by asking specific elements to appear in $H_N$. In~\cite{Webster-23}, the authors studied the effect of $U(b)$ over the basis of $Q(C_1, C_2, 0, 0)$ to characterize $b$ as an element of a kernel (in the case $U(b)$ is acting as a logical identity~\cite[Algorithm 1]{Webster-23}) or an element in the intersection of some spaces \cite[Algorithm 4]{Webster-23}. Here, we compute a matrix for which the kernel is $H_N$. For a fixed $N=2^\ell$, this matrix has at most $k_2\sum_{i=0}^{\ell-1}\binom{k_1}{i}=O(k_2k_1^{\ell-1})$ rows, where $k_j=\dim C_j$, since only products of an element of $\beta_2$ with at most $\ell-1$ elements of $\beta_1$ are involved. Hence, for a fixed $\ell$, deciding whether $b\in H_N$ requires a number of operations that is polynomial in $n$, and a set of generators of $H_N$ can be computed in polynomial time using linear algebra over $\mathbb{Z}_N$. The complexity becomes exponential in $k_1$ only if $\ell$ is allowed to grow with $k_1$.
\end{remark}

The level defined by $2^{\ell-1}\beta_2\star\beta_1^{(\ell-1)}$ gives linear conditions on the codes $C_2$ and $C_1$ that can be used to characterize codes with certain gates.

\begin{corollary}\label{coro.ones}
Let $C_2\subseteq C_1\subseteq\mathbb{F}_2^n$ be linear codes and $Q=Q(C_1,C_2)$ the corresponding CSS code. Then $(1,\ldots,1)\in H_N$ if and only if $N\mid wt(v)$ and $\frac{N}{2}\mid wt(v\star w)$ for all $v\in C_2$ and $w\in C_1$. In particular, if $(1,\ldots,1)\in H_N$, then $C_2\subseteq (C_1^{\ell-1})^{\perp}$.
\end{corollary}

\begin{proof}
    By Equation~(\ref{eq.in}), $(1,\ldots,1)\in H_N$ if and only if $wt(v)-2\,wt(v\star w)=0\bmod N$ for all $v\in C_2$ and $w\in C_1$. Taking $w=0$, we see that this is equivalent to $N\mid wt(v)$ and $\frac N2\mid wt(v\star w)$ for all $v\in C_2$ and $w\in C_1$. For the last claim, since $(1,\ldots,1)\in H_N$, in particular $2^{\ell-1}(1,\ldots,1)\cdot v=2^{\ell-1}wt(v)=0\bmod N$ for any $v\in C_2\star C_1^{\ell-1}$. This implies that $C_2\star (C_1^{\ell-1})$ is even, and we have the conclusion.
\end{proof}

To understand the whole set of diagonal transversal gates, it is required to consider the entire set $\bigcup_{i=1}^\infty \beta_1^{(i)}$. For a generic code, this information is equivalent to knowing every codeword of $C_1$: by Equation~(\ref{eq.2tom}), the lift to $\mathbb{Z}^n$ of $\sum_{w\in S}w$, $S\subseteq\beta_1$, is $\sum_{\emptyset\neq J\subseteq S}(-2)^{|J|-1}\bigstar_{w\in J}w$, and, conversely, the products $\bigstar_{w\in S}w$ can be recovered from these codewords by inclusion-exclusion. For a fixed $N=2^\ell$, only $\bigcup_{i=0}^{\ell-1}\beta_1^{(i)}$ is needed, see Remark~\ref{rem.complexity}.
As we will see, highly structured codes such as monomial codes are more amenable to analysis. Recall that a monomial code $C \subseteq \F_2^n, n=2^m$, is a linear code given by the image of an evaluation map $ev:\mathcal L \rightarrow \F_2^n$ defined by $ev(f):=\left( f(P_1), \dots, f(P_n) \right)$, where $\mathcal L$ is a vector space generated by monomials and $P_1, \dots, P_n \in \F_2^m$ are evaluation points.

\begin{example}\label{ex:rm}
    Let $C_2=RM(0,4)$ and $C_1=RM(1,4)$ be Reed-Muller codes of length $n=2^m=16$. A minimal generating set of $(H_N)^{\perp_N}$ is $\{2^{|\mathbf{i}|}ev(x^\mathbf{i}) :\mathbf{i}\in\{0,1\}^4,\;|\mathbf{i}|\leq \ell-1\}$ (the integer $2^{|\mathbf{i}|}$ times the binary vector $ev(x^\mathbf{i})$, seen in $\mathbb{Z}_N^n$), where $x^\mathbf{i}=x_1^{i_1}\cdots x_4^{i_4}\in\mathbb{F}_2[x_1,x_2,x_3,x_4]$ and for $f$ a polynomial, $ev(f)=(f(P_1),\ldots,f(P_{16}))$, where $\mathbb{F}_2^4=\{P_1,\ldots,P_{16}\}$. 
\end{example}

In fact, for such codes, it is easy to compute the group $H_N$. Let $m\geq 2$ be an integer and $n=2^m$. Fix an order on $\mathbb{F}_2^m$, that is, $\mathbb{F}_2^m=\{P_1,\ldots,P_n\}$, and for any $f\in\mathbb{F}_2[x_1,\ldots,x_m]$, define $ev(f)=(f(P_1),\ldots,f(P_n))\in\mathbb{F}_2^n$. A monomial $u$ is said to be square-free if every variable that appears in $u$ has degree at most 1.  For a monomial $u\in\mathbb{F}_2[x_1,\ldots,x_m]$, $\overline{u}$ represents the square-free monomial supported in the same variables as $u$. Given two sets of square-free monomials $\mathcal{M}_2\subseteq\mathcal{M}_1\subseteq\mathbb{F}_2[x_1,\ldots,x_m]$, we define 
$$
\mathcal M_2\mathcal M_1^t:=\left\{\overline{u_0u_1\cdots u_s}:0\leq s\leq t,u_0\in\mathcal M_2,u_1,\ldots,u_s\in\mathcal M_1\right\}.
$$
For a square-free monomial $w$, let $h(w):=\min\{s\geq 0 :w\in\mathcal{M}_2\mathcal{M}_1^{s}\}$, where $h(w)=\infty$ if there is no such $s$. For a square-free monomial $v$, let $e_v\in\mathbb{Z}_N^n$ be the standard basis vector corresponding to the point $P_v$ of $\mathbb{F}_2^m$ whose nonzero coordinates are those of the variables dividing $v$, and let $\delta_u=\sum_{v\mid u}(-1)^{\deg u-\deg v}e_v$. Since $ev(w)\cdot e_v=w(P_v)=1$ if $w\mid v$ and $0$ otherwise, we get that $\delta_u\cdot ev(w)=1$ if $u=w$ and $\delta_u\cdot ev(w)=0$ otherwise. Hence, $\{ev(u)\}$ and $\{\delta_u\}$, where $u$ runs over the square-free monomials, are dual bases of $\mathbb{Z}_N^n$. The following result determines $H_N$ for any pair of sets of square-free monomials.

\begin{theorem}\label{thm.general}
Let $m\geq 2$ be an integer and let $\mathbb{F}_2^m=\{P_1,\ldots,P_n\}$. Let $\mathcal{M}_2\subseteq\mathcal{M}_1\subseteq\mathbb{F}_2[x_1,\ldots,x_m]$ be two sets of square-free monomials. Let $C_j=\mathrm{Span}_{\mathbb{F}_2}\{ev(u) :u\in\mathcal{M}_j\}$, for $j =1, 2$. Then, for $Q=Q(C_1,C_2)$ and $N=2^\ell$, we have that $b\in H_N$ if and only if $2^{h(w)}b\cdot ev(w)=0\bmod N$ for every square-free monomial $w$ (with $2^\infty=0$). Equivalently,
$$H_N=\mathrm{Span}_{\mathbb{Z}_N}\left\{2^{\ell-\min\{h(w),\ell\}}\delta_w :w \text{ square-free}\right\},$$
and $\lambda(H_N)=\sum_w\min\{h(w),\ell\}$.
\end{theorem}

\begin{proof}
    Since $ev(u)\star ev(v)=ev(uv)$, Theorem~\ref{thm.main} states that $b\in H_N$ if and only if $2^s b\cdot ev(u_0u_1\cdots u_s)=0\bmod N$ for all $u_0\in\mathcal{M}_2$ and distinct $u_1,\ldots,u_s\in\mathcal{M}_1$, $0\leq s\leq\ell-1$. For a fixed $w$, the strongest of these conditions with $\overline{u_0u_1\cdots u_s}=w$ is the one with $s=h(w)$ (repeated factors can be removed), and there is none if $h(w)\geq\ell$. This proves the first claim. Now note that we can write $b=\sum_w(b\cdot ev(w))\delta_w$. Since $2^{h(w)} b\cdot ev(w)=0\bmod N$ for every $w$, then $b\cdot ev(w)\in 2^{\ell-\min\{h(w),\ell\}}\mathbb{Z}_N$ and $b=\sum_w a_w 2^{\ell-\min\{h(w),\ell\}}\delta_w$, $a_w\in \mathbb{Z}_N$. This proves the expression for $H_N$, and $H_N\cong \bigoplus_w 2^{\ell -\min \{h(w),\ell\}}\mathbb{Z}_N$, which gives the length.
\end{proof}

\begin{example}\label{ex.cube}\rm
Since $b\cdot ev(w)=\sum_{P:\,w(P)=1}b_P$, Theorem~\ref{thm.general} says that $U(b)$ fixes $Q$ if and only if, for every square-free monomial $w$ with $h(w)<\ell$, the sum of the entries of $b$ over the subcube $\{P\in\mathbb{F}_2^m :w(P)=1\}$ is divisible by $2^{\ell-h(w)}$. For instance, let $\mathcal{M}_2=\{1\}$ and $\mathcal{M}_1=\{1,x_1,\ldots,x_m\}$, i.e., $C_2=RM(0,m)$ and $C_1=RM(1,m)$. Here $h(w)=\deg w$, so
$$H_N=\mathrm{Span}_{\mathbb{Z}_N}\left\{2^{\ell-\min\{\deg w,\ell\}}\delta_w :w \text{ square-free}\right\}\text{ and }\lambda(H_N)=\sum_{j=0}^{m}\binom{m}{j}\min\{j,\ell\}.$$
$H_N^{\perp_N}$ is generated by the elements $2^{\deg w}ev(w)$ with $\deg w\leq\ell-1$, as in Example~\ref{ex:rm}. In particular, if $\ell=m$, then $\delta_{x_1\cdots x_m}=\sum_{v}(-1)^{m-\deg v}e_v\in H_N$, where $v$ runs over the square-free monomials. Since $P_v$ has $\deg v$ nonzero coordinates, this means that the gate applying $U(1)$ to the qubits indexed by the points with an even number of zero coordinates, and $U(-1)$ to the remaining ones, fixes $Q$. For $m=2$, this is $S$ on the points $00$ and $11$ and $S^\dagger$ on $01$ and $10$; for $m=3$, it is $T$ on the vertices of the cube of odd weight and $T^\dagger$ on those of even weight.
\end{example}

A set of monomials $\mathcal M \subseteq \F_2[x_1, \dots, x_m]$ is closed under divisibility (also called decreasing) if $u \in \mathcal M$ and $v \mid u$ implies $v \in \mathcal M$, meaning $\mathcal M$ contains every factor of each of its elements. An important family of monomial codes is the set of decreasing codes, which are codes generated by monomials closed under divisibility. This family includes polar and Reed-Muller codes, for example. The following result, analogous to the one in~\cite{dualev}, determines the group $H_N$ explicitly for decreasing monomial sets.

\begin{theorem}\label{thm.mainmon}
 Let $\mathcal{M}_2\subseteq\mathcal{M}_1\subseteq\mathbb{F}_2[x_1,\ldots,x_m]$ be two sets of square-free monomials closed under divisibility, such that $x_1,\ldots,x_m\in\mathcal{M}_1$. Define the codes $C_j=\mathrm{Span}_{\mathbb{Z}_2}\{ev(u) :u\in\mathcal{M}_j\}$, for $j=1,2$. Then, for the CSS code $Q=Q(C_1,C_2)$ and $N=2^\ell$, we have
\begin{align*}    H_N&=\mathrm{Span}_{\mathbb{Z}_N}\left(2^{\ell-1-i}ev\left(\frac{x_1\cdots x_m}{u}\right) :0\leq i\leq\ell-1,\ u\notin\mathcal{M}_2\mathcal{M}_1^{i}\ \text{and}\left( u\in\mathcal{M}_2\mathcal{M}_1^{i+1}\ \text{or}\ i=\ell-1\right)\right)\\
&=\mathrm{Span}_{\mathbb{Z}_N}\left(ev\left(\frac{x_1\cdots x_m}{u}\right) :u\notin\mathcal{M}_2\mathcal{M}_1^{\ell-1}\right)+2H_{N/2}\otimes\mathbb{Z}_{N},
\end{align*}
where $2H_{N/2}\otimes\mathbb{Z}_{N}=\{2b :b\in H_{N/2}\}\subseteq\mathbb{Z}_N^n$.
\end{theorem}

\begin{proof}
    Let $u,w$ be square-free monomials, and let $d$ be the number of variables dividing $u$ but not $w$, so that $ev\left(\frac{x_1\cdots x_m}{u}\right)\cdot ev(w)=wt\left(ev\left(w\frac{x_1\cdots x_m}{u}\right)\right)=2^{d}$. Multiplying $w$ by these $d$ variables, which are in $\mathcal{M}_1$, we obtain a multiple of $u$. Since the sets $\mathcal{M}_2\mathcal{M}_1^{s}$ are closed under divisibility (as $\mathcal{M}_1$ and $\mathcal{M}_2$ are), we get $h(u)\leq h(w)+d$, and thus $2^{h(w)}\left(2^{\ell-\min\{h(u),\ell\}}ev\left(\frac{x_1\cdots x_m}{u}\right)\right)\cdot ev(w)=0\bmod N$. By Theorem~\ref{thm.general}, $2^{\ell-\min\{h(u),\ell\}}ev\left(\frac{x_1\cdots x_m}{u}\right)\in H_N$. Since the vectors $ev\left(\frac{x_1\cdots x_m}{u}\right)$ form a basis of $\mathbb{Z}_N^n$, the module generated by these elements is the direct sum of the cyclic modules $2^{\ell-\min\{h(u),\ell\}}\mathbb{Z}_N$, each of which has length $\min\{h(u),\ell\}$. Thus, the generated module has length $\sum_u\min\{h(u),\ell\}=\lambda(H_N)$, and these elements generate $H_N$. The first expression is obtained by grouping these generators according to the value of $h(u)$, and the second one follows from the first one, since $2^{\ell-\min\{h(u),\ell\}}=2\cdot 2^{(\ell-1)-\min\{h(u),\ell-1\}}$ if $h(u)\leq\ell-1$.
\end{proof}

We can also determine $H_N$ in terms of the maximal elements of the sets of monomials.

\begin{corollary}\label{c:maxmon}
Let $\mathcal{M}_2\subseteq\mathcal{M}_1\subseteq\mathbb{F}_2[x_1,\ldots,x_m]$ be two sets of square-free monomials closed under divisibility, with $x_1,\dots,x_m\in \mathcal{M}_1$. Let $\mathcal{B}_i$ be the maximal elements (with respect to the divisibility order) of $\mathcal{M}_i$, $i=1,2$, and define $C_i=\mathrm{Span}_{\mathbb{Z}_2}\{ev(u):u\in\mathcal{M}_i\}$. Assume that $x_1\cdots x_m \notin \mathcal{M}_2\mathcal{M}_1^{\ell-1}$, and let $\Delta$ be the set of square-free monomials that are not multiples of any element of $\left\{\frac{x_1\cdots x_m}{u} :u\in\mathcal{B}_2\mathcal{B}_1^{\ell-1}\right\}$. We have that for $Q=Q(C_1,C_2)$,
    $$H_N=\mathrm{Span}_{\mathbb{Z}_N}\{ev(u) :u\in\Delta\}+2H_{N/2}\otimes \mathbb{Z}_N.$$
\end{corollary}

\begin{proof}
Set $\mathcal D=\mathcal M_2\mathcal M_1^{\ell-1}$, $\mathcal B=\mathcal B_2\mathcal B_1^{\ell-1}$ and $p=x_1\cdots x_m$. Every element of $\mathcal D$ divides some element of $\mathcal B$: replace each factor in its defining product by a maximal element of the corresponding monomial set. Conversely, $\mathcal B\subseteq\mathcal D$, and $\mathcal D$ is closed under divisibility. Thus, for every square-free monomial $v$,
\[
v\in\mathcal D
\quad\Longleftrightarrow\quad
v\mid u\text{ for some }u\in\mathcal B.
\]
This implies
\[
\frac{p}{v}\in\Delta
\quad\Longleftrightarrow\quad
v\notin\mathcal D.
\]
The conclusion follows from the second expression in
Theorem~\ref{thm.mainmon}.
\end{proof}

When we have pairs of Reed-Muller codes, the result is even more transparent.

\begin{corollary}\label{c:rm}   
 Let $C_2=RM(q,m)\subseteq C_1=RM(r,m)$ be Reed-Muller codes (so $0\leq q\leq r$), $r\geq1$, $\ell$ such that $q+(\ell-1)r\leq m-1$, and $N=2^\ell$. Then $\mathrm{H}_N$ is generated by $$\{2^jev(x^{\mathbf{i}}): 0\leq j\leq \ell-1,\; \mathbf{i}\in\{0,1\}^m,\,\, |\mathbf{i}|\leq m-q-(\ell-1-j)r-1\}.$$
\end{corollary}
\begin{proof}
Here $\mathcal{M}_2$ and $\mathcal{M}_1$ are the sets of square-free monomials of degree at most $q$ and $r$, respectively, so $\mathcal{M}_2\mathcal{M}_1^{s}$ consists of the square-free monomials of degree at most $q+sr$. By Theorem~\ref{thm.general} (taking $w=u$) and the proof of Theorem~\ref{thm.mainmon}, for $0\leq j\leq\ell-1$ we have $2^jev\left(\frac{x_1\cdots x_m}{u}\right)\in H_N$ if and only if $h(u)\geq\ell-j$, i.e., $u\notin\mathcal{M}_2\mathcal{M}_1^{\ell-1-j}$, i.e., $\deg u\geq q+(\ell-1-j)r+1$. Writing $x^{\mathbf{i}}=\frac{x_1\cdots x_m}{u}$, so that $|\mathbf{i}|=m-\deg u$, this is the condition in the statement. Since these elements include the generators $2^{\ell-\min\{h(u),\ell\}}ev\left(\frac{x_1\cdots x_m}{u}\right)$ of $H_N$ obtained in the proof of Theorem~\ref{thm.mainmon}, they generate $H_N$.
\end{proof}

The preceding corollary recovers the classification of subcube operators obtained by Barg et al.~\cite[Theorem 2]{Barg}. Indeed, assume $0\leq q<r$, write $\ell=k+1$, and take $\omega=e^{\pi\mathrm{i}/2^k}$. For $\mathbf{i}\in\{0,1\}^m$, the operator $U(ev(x^{\mathbf{i}}))$ applies $Z(k)=\mathrm{diag}(1,\omega)$ to the qubits indexed by the subcube 
\[
A_{\mathbf{i}}=\{(p_1,\ldots,p_m)\in\mathbb{F}_2^m:p_j=1 \text{ whenever } i_j=1\},
\]
whose dimension is $d=m-|\mathbf{i}|$. The corollary shows that this operator preserves $Q$ whenever $d\geq q+kr+1$. Conversely, if $d\leq q+kr$, then $w=(x_1\cdots x_m)/x^{\mathbf{i}}$ satisfies $h(w)\leq k$ and $ev(x^{\mathbf{i}})\cdot ev(w)=1$, so Theorem~\ref{thm.general} shows that the operator does not preserve $Q$.
Similarly, since $Id_N(Q)=H_N(Q(C_1,C_1))$ (see Theorem~\ref{thm.tdmon} below), the same argument with $q$ replaced by $r$ shows that it acts as a logical identity if and only if $d\geq(k+1)r+1$. Since Reed--Muller codes are invariant under translations, these statements hold for every subcube. Thus, the operator implements nontrivial logical action precisely when $q+kr+1\leq d\leq(k+1)r$, recovering both parts of \cite[Theorem 2]{Barg}. Moreover, the preceding corollary provides generators for the entire group $H_N$ of diagonal transversal operators with single-qubit phases that are $N$th roots of unity.

\begin{example}\label{ex.dec}\rm
Let $\mathcal{M}_2=\{1\}$, $\mathcal{M}_1=\{1,x_1,x_2,x_3,x_4,x_1x_2\}\subseteq\mathbb{F}_2[x_1,x_2,x_3,x_4]$, $C_2=\{ev(0),ev(1)\}$ and $C_1$ the span of $\{ev(u) :u\in\mathcal{M}_1\}$. Let $\ell=3$ and observe that $\mathcal{M}_1$ is closed under divisibility with maximal elements, meaning those in the set $\{x_3,x_4,x_1x_2\}$. Thus $\mathcal{M}_2\mathcal{M}_1^{2}$ is closed under divisibility with maximal elements, meaning those monomials in the set $\{x_3x_4,x_1x_2x_3,x_1x_2x_4\}$, and after dividing $x_1x_2x_3x_4$ by each of the elements in the last set, we get $\{x_1x_2,x_4,x_3\}$. The set of monomials $\Delta$ that are not multiples of those in the previous set is $\{1,x_1,x_2\}$. Then
$$H_8=\mathrm{Span}_{\mathbb{Z}_8}\{ev(1),ev(x_1),ev(x_2)\}+2H_4\otimes \mathbb{Z}_8.$$

Similarly,
$$H_4=\mathrm{Span}_{\mathbb{Z}_4}\{ev(u) :u\text{ is a monomial},\, \deg u\leq 2,\, u\neq x_3x_4\}+2H_2\otimes \mathbb{Z}_4,$$
\[
H_2=\mathrm{Span}_{\mathbb{Z}_2}
\left\{
ev(u) :\
u\text{ is a square-free monomial},\
u\neq x_1x_2x_3x_4
\right\}.
\]
The vectors $b\in H_2$ parametrize the diagonal Pauli operators $Z(b)$ that preserve $Q$.
\end{example}

We can, in fact, determine the logical action with the same information we have used to compute the group $H_N$. 

\begin{proposition}
    Assume the same hypothesis as in Theorem \ref{thm.general}. Let $\mathcal{M}_1\setminus\mathcal{M}_2=\{u_1,\ldots,u_K\}$, consider the encoding $\{ev(u_1),\ldots,ev(u_K)\}$, and let $u$ be a square-free monomial such that $ev(u)\in H_N$. Then
    $$
\overline{U(ev(u))}=\sum_{v\in\mathbb{F}_2^K}\omega^{wt(ev(uf_v))}|v\rangle\langle v|,
    $$
    where $f_v=\sum_{i=1}^K v_iu_i$.
\end{proposition}

\begin{proof}
    By Lemma \ref{Lemma.action}, the action of $U(b)$ is determined by $b\begin{bmatrix} |&&|\\ w_1&\cdots&w_K\\|&&|\end{bmatrix} v$. In the case of monomial codes, $\begin{bmatrix} |&&|\\ w_1&\cdots&w_K\\|&&|\end{bmatrix} v$ corresponds to $ev(f_v)$. Since $b=ev(u)$, then 
    $$
    b\cdot ev(f_v)= ev(u)\cdot ev(f_v)=wt(ev(uf_v))\bmod N.
    $$
\end{proof}

We now return to general CSS codes since we have not yet described the logical action realized by our transversal gates. First, we introduce more notation.

\medskip

Let $\omega$ be an $N$-th root of unity where $N=2^\ell$, and for $a\in\mathbb{Z}_N$, let $U(a)=\mathrm{diag}(1,\omega^a)$ as before. Let $n$ be a positive integer and $J\subseteq [n]$. For $v\in\mathbb{Z}_2^n$, we write $v_J=\prod_{i\in J} v_i$. Observe that $v_J=1$ if and only if $v_i=1$ for all $i\in J$. A controlled $U(a)$ is an operator of the form
$$C_J\text{-}U(a):=\sum_{\substack{v\in\mathbb{Z}_2^n\\ v_J=0}}|v\rangle\langle v|+\omega^a\sum_{\substack{v\in\mathbb{Z}_2^n\\ v_J=1}}|v\rangle\langle v|.$$
It is a controlled $U(a)$ in the sense that the application of $U(a)$ to the qubit indexed by the maximum of $J$ is controlled by the subsystem of qubits indexed by the rest of the elements in $J$. Observe that if $J=\{j\}$, then $C_J\text{-}U(a)=U(ae_j)$, where $e_j$ is the $j$-th vector of the standard basis of $\mathbb{Z}_N^n$.

\begin{proposition}\label{p:logact}
    Let $C_2\subseteq C_1\subseteq\mathbb{F}_2^n$ be two linear codes with bases $\beta_2\subseteq\beta_1$, respectively. Assume $\beta_1\setminus\beta_2=\{w_1,\ldots,w_K\}$ and $b\in H_N$. For $J=\{j_1,\ldots,j_h\} \subseteq [K]$, define 
    $$a_{J,b}=(-2)^{h-1}
b\cdot\left(w_{j_1}\star\cdots\star w_{j_h}\right)\bmod N.$$
Then, the logical action of $U(b)$ is given by $$\overline{U(b)}=\prod_{h=1}^K \prod_{\substack{J\subseteq [K]\\ |J|=h}} C_J\text{-}U(a_{J,b}).$$
\end{proposition}

\begin{proof}
    Let $v\in\mathbb{F}_2^K$. Recall that $|v\rangle_L=\left|C_2+\sum_{i=1}^K v_iw_i\right\rangle$ and
    $$U(b)|v\rangle_L=\omega^{b\cdot\left(\sum_{i=1}^K v_iw_i\right)}|v\rangle_L.$$
    Applying Equation~(\ref{eq.2tom}) repeatedly (by induction on $K$), we have $$\sum_{i=1}^K v_iw_i=\bigoplus_{h=1}^K \bigoplus_{\substack{J\subseteq [K]\\ |J|=h}}(-2)^{h-1}\left(\prod_{i\in J} v_i\right)\substack{\bigstar\\ {i\in J}}w_i,$$

    \noindent from which we get $$b\cdot\left(\sum_{i=1}^K v_iw_i\right)=\bigoplus_{h=1}^K \bigoplus_{\substack{J\subseteq [K]\\ |J|=h}}\left(\prod_{i\in J} v_i\right) a_{J,b}.$$

    Since the gate $\sum_{v\in\mathbb{Z}_2^K} \omega^{\left(\prod_{i\in J} v_i\right) a_{J,b}}|v\rangle\langle v|$ is in fact $C_J\text{-}U(a_{J,b})$, we have the conclusion.
\end{proof}

\begin{example}\label{ex:logdec}
    We continue with Example \ref{ex.dec}. Let $w_i=ev(x_i)$, $1\leq i\leq 4$, and $w_5=ev(x_1x_2)$. Fix $\ell=3$. We have that the following gates implement the logical actions:
    $$\begin{array}{ccl}
    U(ev(1))&\mapsto& Z(e_5)\ C_{3,5}\text{-}Z\ C_{4,5}\text{-}Z\ C_{3,4,5}\text{-}Z,\\
    U(ev(x_1))&\mapsto& Z(01111)\ C_{2,3}\text{-}Z\ C_{3,5}\text{-}Z\ C_{2,4}\text{-}Z\ C_{3,4}\text{-}Z\ C_{4,5}\text{-}Z\ C_{3,4,5}\text{-}Z\ C_{2,3,4}\text{-}Z,\\
    U(ev(x_2))&\mapsto& Z(10111)\ C_{1,3}\text{-}Z\ C_{3,5}\text{-}Z\ C_{1,4}\text{-}Z\ C_{3,4}\text{-}Z\ C_{4,5}\text{-}Z\ C_{3,4,5}\text{-}Z\ C_{1,3,4}\text{-}Z.
    \end{array}$$
\end{example}

\begin{example}\label{ex.cubelog}
We continue with Example~\ref{ex.cube}, with $\ell=m$, i.e., $N=2^m$, and the encoding $w_i=ev(x_i)$, $1\leq i\leq m$. For $b=\delta_{x_1\cdots x_m}$ and $J\subseteq[m]$, we have $b\cdot\left(\bigstar_{i\in J}w_i\right)=\delta_{x_1\cdots x_m}\cdot ev\left(\prod_{i\in J}x_i\right)$, which is $1$ if $J=[m]$ and $0$ otherwise. Thus, $a_{J,b}=0$ for $J\neq[m]$, and $a_{[m],b}=(-2)^{m-1}=2^{m-1}\bmod 2^m$, so, by Proposition~\ref{p:logact}, $U(\delta_{x_1\cdots x_m})$ implements the multi-controlled $Z$ gate $C_{[m]}\text{-}Z$, where $Z=U(2^{m-1})$. This is the logical $CZ$ of the $[[4,2,2]]$ code ($m=2$) and the logical $CCZ$ of the $[[8,3,2]]$ code ($m=3$). For $m=3$, the gate $T^{\otimes 8}=U(ev(1))$ also fixes $Q$, and since $wt(ev(x_i))=4$, $wt(ev(x_ix_j))=2$ and $wt(ev(x_1x_2x_3))=1$, it implements
$$Z(111)\ C_{1,2}\text{-}Z\ C_{1,3}\text{-}Z\ C_{2,3}\text{-}Z\ C_{1,2,3}\text{-}Z.$$
\end{example}

In some cases, it is of interest to compute transversal logical operators arising from transversal physical ones, even if it is just a logical identity. Our previous computations show that this case is straightforward to express.

\begin{proposition}\label{prop.trans}
Let $C_2\subseteq C_1\subseteq\mathbb{F}_2^n$ be two linear codes  with bases $\beta_2\subseteq\beta_1$, respectively, and $Q=Q(C_1,C_2)$ the corresponding CSS code. If $N=2^\ell$, then
\[
T_N=\left(\beta_2\cup\bigcup_{i=2}^{\ell}\left(2^{i-1}\beta_1^{(i)}\right)\right)^{\perp_N}=H_N\cap\left(\bigcup_{i=2}^{\ell}\left(2^{i-1}(\beta_1\setminus\beta_2)^{(i)}\right)\right)^{\perp_N},
\]
and     $$Id_N=\left(\bigcup_{i=1}^{\ell}\left(2^{i-1}\beta_1^{(i)}\right)\right)^{\perp_N}=T_N\cap (\beta_1\setminus\beta_2)^{\perp_N}.$$
\end{proposition}

\begin{proof}
Let $b\in H_N$ and $\beta_1\setminus\beta_2=\{w_1,\ldots,w_K\}$. By Lemma~\ref{Lemma.action} and the proof of Proposition~\ref{p:logact}, the logical action of $U(b)$ is $\sum_{v}\omega^{f(v)}|v\rangle\langle v|$, where
$$f(v)=b\cdot\left(\sum_{i=1}^K v_iw_i\right)=\sum_{\emptyset\neq J\subseteq[K]}a_{J,b}\prod_{i\in J}v_i.$$
The logical operator is a tensor product of single-qubit gates if and only if
\[
f(v)=\sum_{i=1}^K v_i f(e_i)\bmod N
\]
for every $v\in\mathbb F_2^K$. Indeed, since the operator is diagonal and $f(0)=0$, this condition is equivalent to
\[
\sum_v\omega^{f(v)}|v\rangle\langle v|
=
\bigotimes_{i=1}^K
\operatorname{diag}(1,\omega^{f(e_i)}).
\]
Since $f$ is represented by a multilinear polynomial and $f(e_i)=a_{\{i\},b}$, this holds if and only if $a_{J,b}=0\bmod N$ for every $J\subseteq[K]$ with $|J|\geq 2$, by uniqueness of multilinear polynomial representations on $\mathbb F_2^K$ over $\mathbb Z_N$. Equivalently,
\[
2^{s-1}b\cdot w_{i_1}\star\cdots\star w_{i_s}
=0\bmod N
\]
for all distinct $i_1,\ldots,i_s\in[K]$ with $s\geq 2$. These conditions are automatic for $s>\ell$, yielding the second expression for $T_N$.

For the first one, let $D=\beta_2\cup\bigcup_{i=2}^{\ell}2^{i-1}\beta_1^{(i)}$. An element of $2^{i-1}\beta_1^{(i)}$ either lies in $2^{i-1}(\beta_1\setminus\beta_2)^{(i)}$ or has a factor $u\in\beta_2$, and then it lies in $2^{i-1}\beta_2\star\beta_1^{(i-1)}$. Since also $\beta_2=\beta_2\star\beta_1^{(0)}$, Theorem~\ref{thm.main} and the second expression give $T_N\subseteq D^{\perp_N}$. Conversely, let $1\leq i\leq\ell-1$ and $2^iu\star w\in 2^i\beta_2\star\beta_1^{(i)}$. If $u$ is not a factor of $w$, then $2^iu\star w\in 2^i\beta_1^{(i+1)}\subseteq D$. Otherwise, $2^iu\star w=2^iw$ is a multiple of an element of $2^{i-1}\beta_1^{(i)}\subseteq D$ if $i\geq 2$, and of $u\in\beta_2$ if $i=1$. Hence, $D^{\perp_N}\subseteq H_N$ by Theorem~\ref{thm.main}, and $D^{\perp_N}\subseteq T_N$ by the second expression, since $(\beta_1\setminus\beta_2)^{(i)}\subseteq\beta_1^{(i)}$.

Finally, $U(b)$ acts as a logical identity if and only if $f(v)=0$ for every $v$, i.e., if and only if $b\in T_N$ and, moreover, $b\cdot w_i=0\bmod N$ for $1\leq i\leq K$. Since $\beta_1=\beta_2\cup(\beta_1\setminus\beta_2)$ and $\beta_1^{(1)}=\beta_1$, this gives the expressions for $Id_N$.
\end{proof}

\begin{remark}\rm
We could have computed $H_N$ for any $Q(C_1,C_2,y_x,y_z)$. The equations defining $H_N$ and $T_N$ would be the same but multiplying componentwise by $t_0=(1,\ldots,1)-2y_z$. On the other hand, by Corollary~\ref{coro.factor}, $Id_N(Q(C_1,C_2,y_x,y_z))=\{b :b\cdot y_z=0\bmod N,\ b_{y_z}\in Id_N(Q(C_1,C_2))\}$.

Another approach would have been to not start with transversal gates with factors of the form $\mathrm{diag}(1,\omega)$ but $\mathrm{diag}(\omega_1,\omega_2)$ and to compute the groups without ignoring the global phase $\omega_1$. 
\end{remark}

If we require certain elements to be in $T_N$, we immediately impose some conditions on the code. This is a way to obtain a generalization of triorthogonal codes, and that explains why divisible codes have appeared in the literature on such codes (see, for instance, \cite{Haah-18, Hu-22}). Recall that a linear code $C$ is said to be a $D$-divisible code if and only if $D \mid wt(c)$ for all codewords $c \in C$.

\begin{corollary}
Let $C_2\subseteq C_1\subseteq\mathbb{F}_2^n$ be linear codes with bases $\beta_2\subseteq\beta_1$, respectively, and $Q=Q(C_1,C_2)$ the corresponding CSS code. Assume that $\beta_1\setminus\beta_2=\{w_1,\ldots,w_K\}$ and $N=2^\ell$.

     \begin{enumerate}
        \item $(1,\ldots,1)\in T_N$ if and only if $wt\left(u+\sum_{i=1}^K v_iw_i\right)=\sum_{i=1}^K v_i wt(w_i)\bmod N$ for any $u\in C_2$ and $v\in\mathbb{F}_2^K$.

        \item $(1,\ldots,1)\in T_N$ and the logical action is $U(1)^{\otimes K}$ if and only if $wt\left(u+\sum_{i=1}^K v_iw_i\right)=\sum_{i=1}^K v_i\bmod N$ for any $u\in C_2$ and $v\in\mathbb{F}_2^K$.

        \item $(1,\ldots,1)\in Id_N$ if and only if $C_1$ is an $N$-divisible code.
     \end{enumerate}
\end{corollary}
\begin{proof}
    \begin{enumerate}
        \item Since $b=(1,\ldots,1)\in H_N$,  $$wt\left(u+\sum_{i=1}^K v_iw_i\right)\bmod N=b\cdot\left(u+\sum_{i=1}^K v_iw_i\right)=b\cdot \sum_{i=1}^K v_iw_i.$$ Since $(1,\dots,1)\in T_N$, we also have
        $$b\cdot\sum_{i=1}^K v_iw_i=b\cdot\bigoplus_{i=1}^K v_iw_i=\sum_{i=1}^K v_iwt(w_i)\bmod N.$$
        Conversely, assume that these equations hold. Any $w\in C_1$ can be written as $w=u'+\sum_{i=1}^K v_iw_i$ with $u'\in C_2$, so $wt(u+w)=\sum_{i=1}^K v_iwt(w_i)=wt(w)\bmod N$ for all $u\in C_2$. Hence, $b\in H_N$, and, taking $u=0$, the logical action of $U(b)$ is $\overline{U(wt(w_1),\ldots,wt(w_K))}$, so $b\in T_N$.

        \item If $b\in T_N$, then the logical action is $\overline{U(wt(w_1),\ldots,wt(w_K))}$, which is $U(1)^{\otimes K}$ just when $wt(w_i)=1\bmod N$ for $1\leq i\leq K$. This gives the result by the first part (for the converse, take $u=0$ and $v=e_i$).

        \item By definition and Lemma~\ref{Lemma.action}, $b\in Id_N$ if and only if $wt(x)=b\cdot x=0\bmod N$ for all $x\in C_1$, i.e., $C_1$ is an $N$-divisible code.
    \end{enumerate}
\end{proof}

\begin{remark}\rm
Compare the characteristic equations defining $T_N$ with the definition of triorthogonal codes for $\ell=3$ \cite{Bravi-12}. For transversal logical $T$ without Clifford correction, compare the second and third parts of the corollary with \cite[Theorem 4, Theorem 3]{Rengaswamy-20}, respectively. The last corollary shows that divisible codes naturally provide transversal gates, which have been studied extensively, for example, \cite{Haah-18, Hu-22}. 
\end{remark}

The choice of a basis of the codes affects the group $T_N$, and it affects the corresponding logical action in the expected way.

\begin{lemma}
Let $C_2\subseteq C_1\subseteq\mathbb{F}_2^n$ be linear codes with bases $\beta_2\subseteq\beta_1$ and $\beta_2'\subseteq\beta'_1$, respectively, and $Q=Q(C_1,C_2)$ be the corresponding CSS code. Assume that $N=2^\ell$, $\beta_1\setminus\beta_2=\{w_1,\ldots,w_K\}$, and $\beta'_1\setminus\beta'_2=\{w'_1,\ldots,w'_K\}$.     Let $L\cong L'\cong (\mathbb{C}^2)^{\otimes K}$ be the space of logical states of $Q(C_1,C_2)$ with the encodings associated to $\beta_1$ and $\beta_1'$ respectively and let $T_N$ and $T'_N$ be the subgroups of $H_N$ that represent transversal logicals for each encoding. Let $M\in\mathrm{GL}_K(\mathbb{F}_2)$ be the change-of-basis matrix determined by $[w'_j]=\sum_{i=1}^K M_{ij}[w_i]$ in $C_1/C_2$. For $b\in T_N\cap T'_N$, if the logical action of $U(b)$ is $\overline{U(c)}$, $c\in\mathbb{Z}_N^K$ over $L$, then the corresponding logical action over $L'$ is $\overline{U(cM)}$. However, $T_N\neq T'_N$ in general.
\end{lemma}

\begin{proof}
For any $v\in\mathbb{F}_2^K$, let $v'=M^{-1}v$ and $w=\sum_{i=1}^K v_iw_i$. Then
$$
C_2+w=C_2+\sum_{i=1}^K v'_iw'_i.
$$
    \noindent Thus, if $b\in T_N\cap T'_N$ and $c=b\begin{bmatrix} |&&|\\ w_1&\cdots&w_K\\|&&|\end{bmatrix}$, then $b\cdot w'_k=c\cdot (Me_k)=(cM)_k$ for $1\leq k\leq K$, and
    \begin{align*}
    U(b)|C_2+w\rangle=\omega^{b\cdot w}|C_2+w\rangle&=\overline{U\left(c\right)}|v\rangle_{L}\\ &=\overline{U\left(cM\right)}|v'\rangle_{L'},
    \end{align*}
    and the result follows. Finally, for $C_2=\{0\}$, $C_1=\mathbb{F}_2^2$, $N\geq 4$, $\beta_1=\{e_1,e_2\}$ and $\beta'_1=\{e_1+e_2,e_2\}$, we have $(0,1)\in T_N$, but $U((0,1))$ acts over $L'$ as $\sum_{v'}\omega^{v'_1+v'_2-2v'_1v'_2}|v'\rangle\langle v'|$, which is not a tensor product, so $(0,1)\notin T'_N$.
\end{proof}

\begin{remark}\rm
In the search for CSS codes with certain transversal gates, we may seek:

\begin{itemize}
    \item Codes with a transversal gate realizing a specific logical action. This is the case for triorthogonal codes, and, according to the previous result, it can be achieved by finding a specific basis that realizes this action.

    \item Codes with a transversal gate fixing the codespace. This is the more general case of CSS-T codes. Since such a result is basis-independent, the papers on such codes discuss the structural properties of the square and its duals.
\end{itemize}

\end{remark}

As in Remark~\ref{rem.complexity}, for a fixed $N=2^\ell$, these groups only involve products of at most $\ell$ elements of $\beta_1$, so they can be computed in polynomial time. However, for general linear codes, computing them for all values of $N$ simultaneously requires the whole set of codewords of $C_1$. For highly structured codes, they can be computed explicitly for all $N$.

\begin{theorem}\label{thm.tdmon}
    Let $m\geq 2$ be an integer and let $\mathbb{F}_2^m=\{P_1,\ldots,P_n\}$. Let $\mathcal{M}_2\subseteq\mathcal{M}_1\subseteq\mathbb{F}_2[x_1,\ldots,x_m]$ be two sets of square-free monomials of cardinalities $k_2<k_1$, respectively. For $i=1,2$, define  $C_i=\mathrm{Span}_{\mathbb{Z}_2}\{ev(u) :u\in\mathcal{M}_i\}$. Consider the CSS code $Q=Q(C_1,C_2)$ with the encoding $\{ev(u) :u\in\mathcal{M}_1\setminus\mathcal{M}_2\}$, and $N=2^\ell$. Let $h_1(w)=\min\{s\geq 0 :w\in\mathcal{M}_1\mathcal{M}_1^{s}\}$ be the function $h$ for the pair $\mathcal{M}_1\subseteq\mathcal{M}_1$, and let $\mathcal{A}$ be the set of minimal elements (with respect to the divisibility order) of $\mathcal{M}_1$ that are not in $\mathcal{M}_2$. Then $Id_N(Q)=H_N(Q(C_1,C_1))$,
    \begin{align*}
    Id_N&=\mathrm{Span}_{\mathbb{Z}_N}\left\{2^{\ell-\min\{h_1(w),\ell\}}\delta_w :w\text{ square-free}\right\},\\
    T_N&=Id_N+2^{\ell-1}\mathrm{Span}_{\mathbb{Z}_N}\left\{\delta_u :u\in\mathcal{M}_1\setminus\mathcal{M}_2\right\}+\mathrm{Span}_{\mathbb{Z}_N}\left\{\delta_u :u\in\mathcal{A}\right\},
    \end{align*}
    and $T_N/Id_N\cong\mathbb{Z}_2^{k_1-k_2-|\mathcal{A}|}\times\mathbb{Z}_N^{|\mathcal{A}|}$. If, moreover, $\mathcal{M}_1$ and $\mathcal{M}_2$ are closed under divisibility, $k_2>0$, and $x_1,\ldots,x_m\in\mathcal{M}_1$, then $\mathcal{A}=\emptyset$ (so $T_N/Id_N\cong\mathbb{Z}_2^{k_1-k_2}$),
    \begin{align*}
    T_N&=Id_N+2^{\ell-1}\mathrm{Span}_{\mathbb{Z}_N}\left(ev\left(\frac{x_1\cdots x_m}{u}\right) :u\in\mathcal{M}_1\setminus\mathcal{M}_2\right),\\
    \text{ and } \quad Id_N&=\mathrm{Span}_{\mathbb{Z}_N}\left(ev\left(\frac{x_1\cdots x_m}{u}\right) :u\notin \mathcal{M}_1\mathcal{M}_1^{\ell-1}\right)+2Id_{N/2}\otimes\mathbb{Z}_N.
    \end{align*}
\end{theorem}
\begin{proof}
    By Lemma~\ref{Lemma.action}, $b\in Id_N$ if and only if $b\cdot x=0\bmod N$ for all $x\in C_1$, i.e., if and only if $U(b)$ fixes the vector $\sum_{x\in C_1}|x\rangle$, which spans $Q(C_1,C_1)$. Since $U(b)|0\rangle=|0\rangle$, this holds if and only if $U(b)Q(C_1,C_1)=Q(C_1,C_1)$. Hence, $Id_N(Q)=H_N(Q(C_1,C_1))$, and the expressions for $Id_N$ follow from Theorems~\ref{thm.general} and~\ref{thm.mainmon} applied to the pair $\mathcal{M}_1\subseteq\mathcal{M}_1$.

    By Proposition~\ref{prop.trans}, the equations defining $T_N$ are those defining $Id_N$, except for $b\cdot ev(u)=0\bmod N$, $u\in\mathcal{M}_1\setminus\mathcal{M}_2$. Thus, as in the proof of Theorem~\ref{thm.general}, for $w\notin\mathcal{M}_1\setminus\mathcal{M}_2$ the conditions on $b\cdot ev(w)$ are the same for $T_N$ and $Id_N$, while for $u\in\mathcal{M}_1\setminus\mathcal{M}_2$ they are $2^{s-1}b\cdot ev(u)=0\bmod N$, for each expression $u=\overline{u_1\cdots u_s}$ with $2\leq s\leq\ell$ distinct $u_1,\ldots,u_s\in\mathcal{M}_1$. If $u\notin\mathcal{A}$, these conditions are equivalent to $2b\cdot ev(u)=0\bmod N$ (take $u=\overline{uv}$, with $v\in\mathcal{M}_1$ a proper divisor of $u$). If $u\in\mathcal{A}$, there are no such expressions, since every $u_i$ divides $u$. Writing $b=\sum_w(b\cdot ev(w))\delta_w$, and since $2^{\ell-\min\{h_1(u),\ell\}}\delta_u=0$ for $u\in\mathcal{M}_1$, we obtain the expression for $T_N$, and $T_N/Id_N\cong\left(2^{\ell-1}\mathbb{Z}_N\right)^{k_1-k_2-|\mathcal{A}|}\times\mathbb{Z}_N^{|\mathcal{A}|}$.

    Finally, assume the additional hypotheses. Then $1\in\mathcal{M}_2$ is a proper divisor of every $u\in\mathcal{M}_1\setminus\mathcal{M}_2$, so $\mathcal{A}=\emptyset$. As in the proof of Theorem~\ref{thm.mainmon}, $2^{\ell-1}ev\left(\frac{x_1\cdots x_m}{u}\right)=\sum_{u\mid w}2^{\ell-1}\delta_w$. If $u\in\mathcal{M}_1\setminus\mathcal{M}_2$ and $u\mid w$, then $w\notin\mathcal{M}_2$, and $2^{\ell-1}\delta_w\in Id_N$ if $w\notin\mathcal{M}_1$, since then $h_1(w)\geq 1$.

    Thus, modulo $Id_N$, $2^{\ell-1}ev\left(\frac{x_1\cdots x_m}{u}\right)$ equals $2^{\ell-1}\delta_u$ plus a sum of elements $2^{\ell-1}\delta_w$ with $w\in\mathcal{M}_1\setminus\mathcal{M}_2$ a proper multiple of $u$. In particular, $2^{\ell-1}ev\left(\frac{x_1\cdots x_m}{u}\right)\in T_N$, and, by induction on $m-\deg u$, every $2^{\ell-1}\delta_u$, $u\in\mathcal{M}_1\setminus\mathcal{M}_2$, lies in $Id_N+2^{\ell-1}\mathrm{Span}_{\mathbb{Z}_N}\left(ev\left(\frac{x_1\cdots x_m}{w}\right) :w\in\mathcal{M}_1\setminus\mathcal{M}_2\right)$. This gives the last expression for $T_N$.
\end{proof}

\begin{corollary}\label{coro.transdec}
    Consider the hypothesis of Theorem~\ref{thm.tdmon}. Let $\mathcal{M}_1\setminus\mathcal{M}_2=\{u_1,\ldots,u_K\}$ and let $b=ev(u)\in T_N$, where $u$ is a square-free monomial. Then the logical action of $U(b)$ is $\overline{U(c)}$ where $c_i=2^{m-\deg\overline{uu_i}}$, $1\leq i\leq K$.
\end{corollary}

\begin{example}\label{ex:monomial_T/I}
For the codes of Example~\ref{ex.cube}, Theorem~\ref{thm.tdmon} gives $T_N/Id_N\cong\mathbb{Z}_2^m$, generated by the classes of $2^{\ell-1}ev\left(\frac{x_1\cdots x_m}{x_i}\right)$, $1\leq i\leq m$. These are the weight-$2$ operators $Z\left(ev\left(\frac{x_1\cdots x_m}{x_i}\right)\right)$, which act as the logical $Z$ on the $i$-th logical qubit, since $ev\left(\frac{x_1\cdots x_m}{x_i}\right)\cdot ev(x_j)$ is $1$ if $j=i$ and $2$ otherwise. Hence, up to logical identities, the only transversal logical operators in this encoding are logical Pauli $Z$ operators, and the multi-controlled $Z$ gates of Example~\ref{ex.cubelog} lie in $H_N\setminus T_N$. For instance, for $m=\ell=3$, $Id_8=\mathrm{Span}_{\mathbb{Z}_8}\{2ev(1),4ev(x_1),4ev(x_2),4ev(x_3)\}$, so $S^{\otimes 8}$ acts as a logical identity.

Similarly, in Example~\ref{ex.dec}, Theorem~\ref{thm.tdmon} gives
\begin{align*}
Id_8&=2Id_4\otimes\mathbb{Z}_8,\qquad Id_4=\mathrm{Span}_{\mathbb{Z}_4}\{ev(1),ev(x_1),ev(x_2)\}+2Id_2\otimes\mathbb{Z}_4,\\
Id_2&=\mathrm{Span}_{\mathbb{Z}_2}\{ev(u) :u \text{ is a monomial},\ \deg u\leq 2,\ u\neq x_3x_4\},\\
T_8&=Id_8+4\,\mathrm{Span}_{\mathbb{Z}_8}\{ev(x_2x_3x_4),ev(x_1x_3x_4),ev(x_1x_2x_4),ev(x_1x_2x_3),ev(x_3x_4)\}.
\end{align*}
In particular, $T^{\otimes 16}=U(ev(1))\in H_8\setminus T_8$, in agreement with Example~\ref{ex:logdec}, while $S^{\otimes 16}=U(2ev(1))\in Id_8$, so the logical operator implemented by $U(ev(1))$ in Example~\ref{ex:logdec} squares to the identity.
\end{example}

\begin{corollary}
    Let $\mathcal{M}_2\subseteq\mathcal{M}_1\subseteq\mathbb{F}_2[x_1,\ldots,x_m]$ be families of square-free monomials closed under divisibility, with $x_1,\ldots,x_m\in\mathcal{M}_1$. Let $\mathcal{B}_i$ be the maximal elements (with respect to the divisibility order) of $\mathcal{M}_i$, $i=1,2$, and define $C_i=\mathrm{Span}_{\mathbb{Z}_2}\{ev(u):u\in\mathcal{M}_i\}$. Assume that $\mathcal{M}_1\mathcal{M}_1^{\ell-1}$ does not contain $x_1\cdots x_m$ and let $\Delta$ be the square-free monomials that are not a multiple of $\left\{\frac{x_1\cdots x_m}{u} :u\in\mathcal{B}_1\mathcal{B}_1^{\ell-1}\right\}$. We have that for $Q=Q(C_1,C_2)$,
    $$Id_N=\mathrm{Span}_{\mathbb{Z}_N}\left\{ev(u) :u\in\Delta\right\}+2Id_{N/2}\otimes\mathbb{Z}_N.$$
\end{corollary}
\begin{proof}
    By Theorem~\ref{thm.tdmon}, $Id_N(Q)=H_N(Q(C_1,C_1))$, so this is Corollary~\ref{c:maxmon} for the pair $\mathcal{M}_1\subseteq\mathcal{M}_1$.
\end{proof}

\begin{corollary}\label{coro.transrm}
 Let $C_2=RM(q,m)\subseteq C_1=RM(r,m)$ be Reed-Muller codes (so $0\leq q\leq r$), $r\geq1$, $\ell$ such that $\ell r\leq m-1$, and $N=2^\ell$. Then
$$Id_N=\mathrm{Span}_{\mathbb{Z}_N}\{2^jev(x^{\mathbf{i}}) :0\leq j\leq\ell-1,\ \mathbf{i}\in\{0,1\}^m,\,\, |\mathbf{i}|\leq m-(\ell-j)r-1\}.$$
\end{corollary}
\begin{proof}
    By Theorem~\ref{thm.tdmon}, $Id_N(Q)=H_N(Q(C_1,C_1))$, so this follows from Corollary~\ref{c:rm} applied with $q=r$.
\end{proof}

\section{Diagonal gates across several code blocks}\label{sec:multiblock}

In this section, we consider diagonal gates acting transversally on several code blocks, such as the transversal $CCZ$ gate between three blocks or the transversal controlled-$S$ gate between two blocks. We show that the machinery of Section~\ref{sec:diagonal} applies directly to these gates, since they act as single-qubit diagonal gates on a larger CSS code. This reduction uses the controlled-phase/phase-rotation duality and the embedded code method of~\cite[Sections~4.2--4.3]{Webster-23}, where it is used algorithmically; combined with Theorem~\ref{thm.main}, it gives explicit conditions in terms of products of basis vectors and explicit logical actions.

Let $t\geq 1$ and, for $\alpha\in[t]$, let $C_{2,\alpha}\subseteq C_{1,\alpha}\subseteq\mathbb{F}_2^n$ be linear codes with bases $\beta_{2,\alpha}\subseteq\beta_{1,\alpha}$, respectively, and $Q_\alpha=Q(C_{1,\alpha},C_{2,\alpha})$. We index the qubits of $(\mathbb{C}^2)^{\otimes tn}\supseteq Q_1\otimes\cdots\otimes Q_t$ by the pairs $(\alpha,i)\in[t]\times[n]$. For $b\in\mathbb{Z}_N^n$, we define
$$U_t(b)=\prod_{i=1}^nC_{\{(1,i),\ldots,(t,i)\}}\text{-}U(b_i),$$
i.e., $U_t(b)$ applies the $(t-1)$-controlled $U(b_i)$ to the $i$-th qubits of the $t$ blocks. Equivalently, for $x^{(1)},\ldots,x^{(t)}\in\mathbb{F}_2^n$,
$$U_t(b)|x^{(1)}\rangle\cdots|x^{(t)}\rangle=\omega^{b\cdot\left(x^{(1)}\star\cdots\star x^{(t)}\right)}|x^{(1)}\rangle\cdots|x^{(t)}\rangle.$$
\begin{example}
We have that $U_1(b)=U(b)$, and, for $b=(1,\ldots,1)$, $U_2(b)=CS^{\otimes n}$ if $N=4$, and $U_3(b)=CCZ^{\otimes n}$ if $N=2$, where each $CS$ or $CCZ$ gate acts on the $i$-th qubits of the blocks. More generally, the case $N=2$ gives the gate family $C^{t-1}Z$. 
\end{example}

\begin{definition}\rm
With the notation above, we define
\begin{align*}
H_N(Q_1,\ldots,Q_t)&=\{b\in\mathbb{Z}_N^n :U_t(b)(Q_1\otimes\cdots\otimes Q_t)=Q_1\otimes\cdots\otimes Q_t\},\text{ and}\\
Id_N(Q_1,\ldots,Q_t)&=\{b\in H_N(Q_1,\ldots,Q_t) :U_t(b)v=v,\ \forall v\in Q_1\otimes\cdots\otimes Q_t\}.
\end{align*}
If we fix encodings of $Q_1,\ldots,Q_t$ with the same number $K$ of logical qubits, we also define $T_N(Q_1,\ldots,Q_t)$ as the set of $b\in H_N(Q_1,\ldots,Q_t)$ such that $U_t(b)$ is a transversal logical between the blocks, i.e., its logical action is $\bigotimes_{j=1}^KE_j$, where $E_j$ acts on the $j$-th logical qubits of the $t$ blocks. Observe that
$$Id_N(Q_1,\ldots,Q_t)\subseteq T_N(Q_1,\ldots,Q_t)\subseteq H_N(Q_1,\ldots,Q_t).$$
If $Q_1=\cdots=Q_t=Q$, with the same encoding, we write $H_N^{(t)}(Q)$, $T_N^{(t)}(Q)$, and $Id_N^{(t)}(Q)$.
\end{definition}

For $t=1$, these are the groups $H_N(Q_1)$, $T_N(Q_1)$, and $Id_N(Q_1)$ of Section~\ref{sec:diagonal}. As before, $T_N$ depends on the encodings, while $H_N$ and $Id_N$ do not. For $s\geq 0$, let
$$\mathcal{P}'_s=\bigcup_{i_1+\cdots+i_t=s}\beta_{1,1}^{(i_1+1)}\star\cdots\star\beta_{1,t}^{(i_t+1)},$$
where $i_1,\ldots,i_t\geq 0$, and let $\mathcal{P}_s$ be defined in the same way, but replacing, for some $\alpha\in[t]$, the factor $\beta_{1,\alpha}^{(i_\alpha+1)}$ by $\beta_{2,\alpha}\star\beta_{1,\alpha}^{(i_\alpha)}$ (possibly equal to another factor from block $\alpha$). That is, the elements of $\mathcal{P}'_s$ are the products of $t+s$ basis vectors, at least one from each block, where the ones taken from the same block are distinct, and the elements of $\mathcal{P}_s$ have, in addition, a factor in $\beta_{2,\alpha}$ for some $\alpha\in[t]$. For $t=1$, we have $\mathcal{P}'_s=\beta_1^{(s+1)}$ and $\mathcal{P}_s=\beta_2\star\beta_1^{(s)}$. The following result recovers Theorem~\ref{thm.main} and the first expression of $Id_N$ in Proposition~\ref{prop.trans} for $t=1$. 

\begin{theorem}\label{thm.multi}
With the notation above, if $N=2^\ell$, then
$$H_N(Q_1,\ldots,Q_t)=\left(\bigcup_{s=0}^{\ell-1}2^s\mathcal{P}_s\right)^{\perp_N}\quad\text{ and }\quad Id_N(Q_1,\ldots,Q_t)=\left(\bigcup_{s=0}^{\ell-1}2^s\mathcal{P}'_s\right)^{\perp_N}.$$
\end{theorem}
\begin{proof}
Let $\mathcal{S}$ be the set of nonempty subsets of $[t]$, and consider the injective linear map $E:(\mathbb{F}_2^n)^t\rightarrow(\mathbb{F}_2^n)^{|\mathcal{S}|}\cong\mathbb{F}_2^{(2^t-1)n}$ given by $E(x^{(1)},\ldots,x^{(t)})=\left(\sum_{\alpha\in S}x^{(\alpha)}\right)_{S\in\mathcal{S}}$. By Equation~(\ref{eq.2tom}), the lift to $\mathbb{Z}^n$ of $\sum_{\alpha\in S}x^{(\alpha)}$ is
$$
\bigoplus_{\emptyset\neq R\subseteq S}(-2)^{|R|-1}\bigstar_{\alpha\in R}x^{(\alpha)},
$$
and, since $\sum_{S\supseteq R}(-1)^{|S|-1}$ is $(-1)^{t-1}$ if $R=[t]$ and $0$ otherwise, we get
\begin{equation}\label{eq.polar}
\begin{aligned}
\bigoplus_{S\in\mathcal{S}}(-1)^{|S|-1}
\sum_{\alpha\in S}x^{(\alpha)}
&=\bigoplus_{S\in\mathcal{S}}(-1)^{|S|-1}
\bigoplus_{\emptyset\neq R\subseteq S}
(-2)^{|R|-1}\bigstar_{\alpha\in R}x^{(\alpha)}\\
&=\bigoplus_{\emptyset\neq R\subseteq[t]}
(-2)^{|R|-1}
\left(\sum_{S\supseteq R}(-1)^{|S|-1}\right)
\bigstar_{\alpha\in R}x^{(\alpha)}\\
&=2^{t-1}x^{(1)}\star\cdots\star x^{(t)}.
\end{aligned}
\end{equation}
Fix $b\in\mathbb{Z}_N^n$, and let $\hat b\in\{0,\ldots,N-1\}^n$ be its standard integer representative. Set $N'=2^{\ell+t-1}$ and $\omega'=e^{2\pi\mathrm{i}/N'}$, so that $\omega=\omega'^{2^{t-1}}$. Define $\tilde b\in\mathbb{Z}_{N'}^{(2^t-1)n}$ by
$\tilde b_S=(-1)^{|S|-1}\hat b\bmod N'$, for $S\in\mathcal{S}$. By Equation~(\ref{eq.polar}), $\omega^{b\cdot(x^{(1)}\star\cdots\star x^{(t)})}=\omega'^{\tilde b\cdot E(x^{(1)},\ldots,x^{(t)})}$. Let $\tilde C_j=E(C_{j,1}\times\cdots\times C_{j,t})$, for $j=1,2$, and $\tilde Q=Q(\tilde C_1,\tilde C_2)$. By Lemma~\ref{lem.propcss}, the isometry $|x\rangle\mapsto|E(x)\rangle$ maps $Q_1\otimes\cdots\otimes Q_t$ onto $\tilde Q$, and it transforms $U_t(b)$ into the operator $U(\tilde b)$ defined with $N'$ instead of $N$ (that is, if $V|x\rangle=| E(x)\rangle $, then $U(\tilde{b})V=VU_t(b)$). Hence, $b\in H_N(Q_1,\ldots,Q_t)$ if and only if $\tilde b\in H_{N'}(\tilde Q)$, and $b\in Id_N(Q_1,\ldots,Q_t)$ if and only if $\tilde b\in Id_{N'}(\tilde Q)$.

For $\alpha\in[t]$ and $v\in\mathbb{F}_2^n$, let $E_\alpha(v)$ be the image under $E$ of the tuple with $v$ in position $\alpha$ and $0$ elsewhere, i.e., $E_\alpha(v)_S=v$ if $\alpha\in S$, and $E_\alpha(v)_S=0$ otherwise. Then $\tilde\beta_j=\{E_\alpha(v) :\alpha\in[t],\ v\in\beta_{j,\alpha}\}$ is a basis of $\tilde C_j$, $j=1,2$, and $\tilde\beta_2\subseteq\tilde\beta_1$. By Theorem~\ref{thm.main}, $\tilde b\in H_{N'}(\tilde Q)$ if and only if
$$2^{k-1}\tilde b\cdot\left(E_{\alpha_1}(v_1)\star\cdots\star E_{\alpha_k}(v_k)\right)=0\bmod N'$$
for all $k\geq 1$ and all distinct $E_{\alpha_1}(v_1),\ldots,E_{\alpha_k}(v_k)\in\tilde\beta_1$ with $v_1\in\beta_{2,\alpha_1}$ (if the element of $\tilde\beta_2$ is repeated in a product of $2^i\tilde\beta_2\star\tilde\beta_1^{(i)}$, we obtain a multiple of an element of $2^{i-1}\tilde\beta_2\star\tilde\beta_1^{(i-1)}$). By Proposition~\ref{prop.trans}, $\tilde b\in Id_{N'}(\tilde Q)$ if and only if the same holds without the condition $v_1\in\beta_{2,\alpha_1}$. Since the $S$-th component of $E_{\alpha_1}(v_1)\star\cdots\star E_{\alpha_k}(v_k)$ is $v_1\star\cdots\star v_k$ if $\{\alpha_1,\ldots,\alpha_k\}\subseteq S$, and $0$ otherwise, we have
$$\tilde b\cdot\left(E_{\alpha_1}(v_1)\star\cdots\star E_{\alpha_k}(v_k)\right)=\sum_{S\supseteq\{\alpha_1,\ldots,\alpha_k\}}(-1)^{|S|-1}\,\hat b\cdot(v_1\star\cdots\star v_k),$$
which is $(-1)^{t-1}\hat b\cdot(v_1\star\cdots\star v_k)$ if $\{\alpha_1,\ldots,\alpha_k\}=[t]$, and $0$ otherwise. Thus, only the products with factors from every block give nontrivial conditions, and, writing $k=t+s$, these conditions read $2^{t+s-1}\hat b\cdot(v_1\star\cdots\star v_k)=0\bmod 2^{\ell+t-1}$, i.e., $2^s b\cdot(v_1\star\cdots\star v_k)=0\bmod N$, which are trivial if $s\geq\ell$. These are precisely the conditions given by $2^s\mathcal{P}_s$ (resp. $2^s\mathcal{P}'_s$), $0\leq s\leq\ell-1$, where, as before, the products in $\mathcal{P}_s$ in which the factor in $\beta_{2,\alpha}$ is repeated give conditions implied by the other ones.
\end{proof}

The logical action is obtained as in Proposition~\ref{p:logact}.

\begin{proposition}\label{p:logactmb}
Assume the notation of Theorem~\ref{thm.multi}. For $\alpha\in[t]$, let $\beta_{1,\alpha}\setminus\beta_{2,\alpha}=\{w_{\alpha,1},\ldots,w_{\alpha,K_\alpha}\}$ be the encoding of $Q_\alpha$, and index the logical qubits of $Q_1\otimes\cdots\otimes Q_t$ by $\mathcal{K}=\{(\alpha,j) :\alpha\in[t],\ j\in[K_\alpha]\}$. Let $\mathcal{J}$ be the set of subsets $J\subseteq\mathcal{K}$ meeting every block, i.e., such that $J\cap(\{\alpha\}\times[K_\alpha])\neq\emptyset$ for all $\alpha\in[t]$. For $J\in\mathcal{J}$, let $w_J=\bigstar_{(\alpha,j)\in J}w_{\alpha,j}$ and, for $b\in H_N(Q_1,\ldots,Q_t)$,
$$a_{J,b}=(-2)^{|J|-t}\,b\cdot w_J\bmod N.$$
Then, the logical action of $U_t(b)$ is given by
$$\overline{U_t(b)}=\prod_{J\in\mathcal{J}}C_J\text{-}U(a_{J,b}).$$
\end{proposition}

\begin{proof}
For $\alpha\in[t]$, let $v^{(\alpha)}\in\mathbb{F}_2^{K_\alpha}$ and $x^{(\alpha)}=\sum_{j=1}^{K_\alpha}v^{(\alpha)}_jw_{\alpha,j}$. By the proof of Theorem~\ref{thm.multi} and Lemma~\ref{Lemma.action} applied to $\tilde Q$ and $U(\tilde b)$, the operator $U_t(b)$ multiplies $|v^{(1)}\rangle_L\cdots|v^{(t)}\rangle_L$ by $\omega'^{\tilde b\cdot E(x^{(1)},\ldots,x^{(t)})}=\omega^{b\cdot(x^{(1)}\star\cdots\star x^{(t)})}$. As in the proof of Proposition~\ref{p:logact}, the lift of $x^{(\alpha)}$ is $\bigoplus_{\emptyset\neq J_\alpha\subseteq[K_\alpha]}(-2)^{|J_\alpha|-1}\left(\prod_{j\in J_\alpha}v^{(\alpha)}_j\right)\bigstar_{j\in J_\alpha}w_{\alpha,j}$. Taking the componentwise product over $\alpha\in[t]$, and since $\sum_{\alpha}(|J_\alpha|-1)=|J|-t$ for $J=\bigcup_\alpha\{\alpha\}\times J_\alpha$, we get
$$b\cdot\left(x^{(1)}\star\cdots\star x^{(t)}\right)=\bigoplus_{J\in\mathcal{J}}\left(\prod_{(\alpha,j)\in J}v^{(\alpha)}_j\right)a_{J,b},$$
and the conclusion follows as in the proof of Proposition~\ref{p:logact}.
\end{proof}

The following result describes the transversal logicals between the blocks, and recovers the first expression of $T_N$ and the second expression of $Id_N$ in Proposition~\ref{prop.trans} for $t=1$.

\begin{corollary}\label{c:transmb}
Assume the notation of Proposition~\ref{p:logactmb} and $K_1=\cdots=K_t=K$. Then
$$T_N(Q_1,\ldots,Q_t)=\left(\mathcal{P}_0\cup\mathcal{N}\cup\bigcup_{s=1}^{\ell-1}2^s\mathcal{P}'_s\right)^{\perp_N},$$
where $\mathcal{N}=\{w_{1,j_1}\star\cdots\star w_{t,j_t} :j_1,\ldots,j_t\in[K]\text{ not all equal}\}$, and, for $b\in T_N(Q_1,\ldots,Q_t)$, the logical action of $U_t(b)$ is $\prod_{j=1}^KC_{\{(1,j),\ldots,(t,j)\}}\text{-}U(c_j)$, with $c_j=b\cdot(w_{1,j}\star\cdots\star w_{t,j})$. Moreover,
$$Id_N(Q_1,\ldots,Q_t)=T_N(Q_1,\ldots,Q_t)\cap\{w_{1,j}\star\cdots\star w_{t,j} :j\in[K]\}^{\perp_N}.$$
\end{corollary}

\begin{proof}
Let $b\in H_N(Q_1,\ldots,Q_t)$ and, for $j\in[K]$, let $J_j=\{(1,j),\ldots,(t,j)\}$, which is the only element of $\mathcal{J}$ contained in $\{(1,j),\ldots,(t,j)\}$, since every $J\in\mathcal{J}$ meets every block. As in the proof of Proposition~\ref{prop.trans}, by Proposition~\ref{p:logactmb} and the uniqueness of multilinear polynomial representations, $b\in T_N(Q_1,\ldots,Q_t)$ if and only if $a_{J,b}=0\bmod N$ for every $J\in\mathcal{J}\setminus\{J_1,\ldots,J_K\}$, and then the logical action of $U_t(b)$ is $\prod_{j=1}^KC_{J_j}\text{-}U(c_j)$, since $a_{J_j,b}=c_j$. For $|J|=t$, these conditions read $b\cdot w=0\bmod N$, $w\in\mathcal{N}$, and, for $|J|=t+s$ with $s\geq 1$, they read $2^sb\cdot w=0\bmod N$, where $w$ runs over the elements of $\mathcal{P}'_s$ whose factors from each block $\alpha$ lie in $\beta_{1,\alpha}\setminus\beta_{2,\alpha}$ (these conditions are trivial if $s\geq\ell$). The expression for $T_N(Q_1,\ldots,Q_t)$ then follows from Theorem~\ref{thm.multi} as in the proof of Proposition~\ref{prop.trans}, since an element of $\mathcal{P}'_s$ with a factor from some block $\alpha$ in $\beta_{2,\alpha}$ lies in $\mathcal{P}_s$, and, for $s\geq 1$, an element of $2^s\mathcal{P}_s$ either lies in $2^s\mathcal{P}'_s$ or, if its factor in $\beta_{2,\alpha}$ is repeated, is a multiple of an element of $\mathcal{P}_0$ (if $s=1$) or of $2^{s-1}\mathcal{P}'_{s-1}$ (if $s\geq 2$). Finally, $U_t(b)$ acts as a logical identity if and only if, moreover, $c_j=0\bmod N$ for all $j\in[K]$.
\end{proof}

\begin{remark}\rm
\begin{enumerate}
\item For $N=2$, only $s=0$ plays a role and $a_{J,b}=0$ if $|J|>t$. Thus, $U_t(b)$, which applies $C^{t-1}Z$ to the $i$-th qubits of the blocks for every $i$ with $b_i=1$, fixes $Q_1\otimes\cdots\otimes Q_t$ if and only if
$$b\in\left(\sum_{\alpha=1}^tC_{1,1}\star\cdots\star C_{1,\alpha-1}\star C_{2,\alpha}\star C_{1,\alpha+1}\star\cdots\star C_{1,t}\right)^\perp,$$
and, in that case, its logical action is the product of the gates $C_{\{(1,j_1),\ldots,(t,j_t)\}}$-$Z$ with $b\cdot(w_{1,j_1}\star\cdots\star w_{t,j_t})=1\bmod 2$. These are precisely the conditions underlying the transversal $CCZ$ gates of~\cite{golowich_asymptotically_good_quantum_transversal,nguyen_good_css_ccz,he_addressable_transversal}, where one asks for $b\in T_2(Q_1,Q_2,Q_3)$ with $c_j=1$ for all $j\in[K]$ in Corollary~\ref{c:transmb}, i.e., such that $U_3(b)$ implements the logical $CCZ^{\otimes K}$. For instance, if $Q_1=Q_2=Q_3=Q$, where $C_1$ is spanned by all rows of a triorthogonal matrix~\cite{Bravi-12}, $C_2$ is spanned by its even-weight rows, and the encoding is given by its odd-weight rows, then $(1,\ldots,1)\in T_2^{(3)}(Q)$ and $CCZ^{\otimes n}$ implements $\prod_jC_{\{(1,j),(2,j),(3,j)\}}$-$Z=CCZ^{\otimes K}$. Thus, every triorthogonal code admits a transversal $CCZ$ gate
across three identically encoded copies, implementing
$CCZ^{\otimes K}$ on the corresponding logical qubits, without any Clifford correction.
\item By \cite{cui_diagonal_gates}, the gate $C_{\{1,\ldots,t\}}$-$U(1)$, with $N=2^\ell$, lies in the level $\ell+t-1$ of the Clifford hierarchy. This is reflected in the proof of Theorem~\ref{thm.multi}, where $U_t(b)$ acts as a gate $U(\tilde b)$ with $2^{\ell+t-1}$-th roots of unity on a larger code.
\item If $Q_1=\cdots=Q_t=Q$, then the reduction modulo $N$ of $H_{2^{\ell+t-1}}(Q)$ is contained in $H_N^{(t)}(Q)$. Indeed, an element $w\in\mathcal{P}_s$ is then a product of $k\leq t+s$ distinct elements of $\beta_1$, one of them in $\beta_2$, and, for $b\in H_{2^{\ell+t-1}}(Q)$, Theorem~\ref{thm.main} gives $2^{k-1}b\cdot w=0\bmod 2^{\ell+t-1}$ (the condition is trivial if $k>\ell+t-1$, since then $s\geq\ell$), so $2^sb\cdot w=0\bmod N$. This inclusion is an equality for monomial codes (see Corollary~\ref{c:monmb} below).
\end{enumerate}
\end{remark}

For monomial codes, we obtain the multi-block analog of Theorem~\ref{thm.general}.

\begin{corollary}\label{c:monmb}
Let $\mathcal{M}_2\subseteq\mathcal{M}_1$, $C_1$, $C_2$, and $Q=Q(C_1,C_2)$ be as in Theorem~\ref{thm.general}, and let $t\geq 1$. Then $b\in H_N^{(t)}(Q)$ if and only if $2^{\max\{h(w)-t+1,0\}}b\cdot ev(w)=0\bmod N$ for every square-free monomial $w$ (with $2^\infty=0$). Equivalently,
$$H_N^{(t)}(Q)=\mathrm{Span}_{\mathbb{Z}_N}\left\{2^{\ell-\min\{\max\{h(w)-t+1,0\},\ell\}}\delta_w :w \text{ square-free}\right\},$$
which is the reduction modulo $N$ of $H_{2^{\ell+t-1}}(Q)$. In particular, if $C_2=RM(q,m)\subseteq C_1=RM(r,m)$ (so $q\leq r$), $r\geq 1$, and $q+(\ell+t-2)r\leq m-1$, then $H_N^{(t)}(Q)$ is generated by
$$\{2^jev(x^{\mathbf{i}}) :0\leq j\leq \ell-1,\ \mathbf{i}\in\{0,1\}^m,\ |\mathbf{i}|\leq m-q-(\ell+t-2-j)r-1\}.$$
\end{corollary}

\begin{proof}
Since $ev(u)\star ev(v)=ev(uv)$, every element of $\mathcal{P}_s$ is of the form $ev(w)$, where $w$ is the square-free part of a product of $t+s$ monomials of $\mathcal{M}_1$, one of them in $\mathcal{M}_2$, so $h(w)\leq t+s-1$. Conversely, if $w=\overline{u_0u_1\cdots u_{h(w)}}$, with $u_0\in\mathcal{M}_2$ and distinct $u_1,\ldots,u_{h(w)}\in\mathcal{M}_1\setminus\{u_0\}$, distributing these $h(w)+1$ monomials among the $t$ blocks, and using $u_0$ in the blocks that remain empty if $h(w)+1<t$, shows that $ev(w)\in\mathcal{P}_s$ with $s=\max\{h(w)-t+1,0\}$. Hence, by Theorem~\ref{thm.multi}, $b\in H_N^{(t)}(Q)$ if and only if $2^{\max\{h(w)-t+1,0\}}b\cdot ev(w)=0\bmod N$ for every $w$. The expression for $H_N^{(t)}(Q)$ follows as in the proof of Theorem~\ref{thm.general}, and comparing it with the one given by Theorem~\ref{thm.general} for $2^{\ell+t-1}$ shows that it is the reduction modulo $N$ of $H_{2^{\ell+t-1}}(Q)$. The last claim follows from Corollary~\ref{c:rm} with $\ell+t-1$ in place of $\ell$, since the generators with $j\geq\ell$ vanish modulo $N$.
\end{proof}

\begin{example}\label{ex.cubemb}\rm
We continue with Example~\ref{ex.cube} for $m=3$, i.e., $Q$ is the $[[8,3,2]]$ code, with the encoding $w_j=ev(x_j)$, $1\leq j\leq 3$, and $h(w)=\deg w$. For $t=2$ and $N=4$, Corollary~\ref{c:monmb} gives
$$H_4^{(2)}(Q)=\mathrm{Span}_{\mathbb{Z}_4}\{2\delta_{x_1x_2},2\delta_{x_1x_3},2\delta_{x_2x_3},\delta_{x_1x_2x_3}\},$$
a group of order $2^5$. Since $\delta_u\cdot ev(w)=1$ if $u=w$ and $0$ otherwise, Proposition~\ref{p:logactmb} shows that $U_2(2\delta_{x_ix_j})$ implements $C_{\{(1,i),(2,j)\}}\text{-}Z\ C_{\{(1,j),(2,i)\}}\text{-}Z$, where $Z=U(2)$, and that $U_2(\delta_{x_1x_2x_3})$, which applies $CS$ to the pairs of qubits indexed by the points of odd weight and $CS^\dagger$ to the remaining ones, implements the product of the six logical $CCZ$ gates $C_{\{(\alpha,i),(\alpha,j),(\alpha',k)\}}\text{-}Z$, with $\{i,j,k\}=\{1,2,3\}$, $i<j$, and $\{\alpha,\alpha'\}=\{1,2\}$. 
Compare with Example~\ref{ex.cubelog}, where $U(\delta_{x_1x_2x_3})$, with $N=8$, implements a single logical $CCZ$ on one block.
\end{example}

Recently, there has been considerable interest in transversal gates that can be addressed to prescribed logical qubits~\cite{he_addressable_transversal,he_good_addressable,sanjose_good_triorthogonal}. Addressing a single logical gate $C_{J_0}\text{-}U(c)$ amounts to prescribing $c_{J_0}=c$ and $c_J=0$ for every $J\neq J_0$. The following result characterizes the existence of such implementations and their minimum number of nontrivial physical gates. By Proposition~\ref{p:logactmb}, the coefficients $a_{J,b}$ of the logical action of $U_t(b)$ depend linearly on $b$, so these questions reduce to linear algebra over $\mathbb{Z}_N$ (except the question about the minimum number of nontrivial gates).

\begin{corollary}\label{c:address}
Assume the notation of Proposition~\ref{p:logactmb}.
\begin{enumerate}
\item Let $(c_J)_{J\in\mathcal{J}}\in\mathbb{Z}_N^{\mathcal{J}}$. There is $b\in\mathbb{Z}_N^n$ such that $U_t(b)$ fixes $Q_1\otimes\cdots\otimes Q_t$ and implements the logical operator $\prod_{J\in\mathcal{J}}C_J\text{-}U(c_J)$ if and only if the system of linear equations
$$b\in\left(\bigcup_{s=0}^{\ell-1}2^s\mathcal{P}_s\right)^{\perp_N},\qquad (-2)^{|J|-t}\,b\cdot w_J=c_J\bmod N\ \text{ for all } J\in\mathcal{J},$$
has a solution. In that case, the set of solutions is a coset $b_0+Id_N(Q_1,\ldots,Q_t)$, and the minimum number of nontrivial gates $C_{\{(1,i),\ldots,(t,i)\}}\text{-}U(b_i)$ in such an implementation is $\min\{wt(b) :b\in b_0+Id_N(Q_1,\ldots,Q_t)\}$.
\item Let $N=2$ and $J_0=\{(1,j_1),\ldots,(t,j_t)\}$. Then, there is $b\in\mathbb{Z}_2^n$ such that $U_t(b)$ implements the single logical gate $C_{J_0}\text{-}Z$ if and only if $w_{J_0}\notin D$, where
$$D=\mathrm{Span}_{\mathbb{Z}_2}\left(\mathcal{P}_0\cup\{w_J :J\in\mathcal{J},\ |J|=t,\ J\neq J_0\}\right).$$
In that case, the minimum number of physical $C^{t-1}Z$ gates in such an implementation is $\min\{wt(b) :b\in D^\perp\setminus(D+\langle w_{J_0}\rangle)^\perp\}$.
\end{enumerate}
\end{corollary}

\begin{proof}
By Proposition~\ref{p:logactmb} and the uniqueness of multilinear polynomial representations (see the proof of Proposition~\ref{prop.trans}), $U_t(b)$ implements $\prod_{J\in\mathcal{J}}C_J\text{-}U(c_J)$ if and only if $b\in H_N(Q_1,\ldots,Q_t)$ and $a_{J,b}=c_J$ for every $J\in\mathcal{J}$, which, by Theorem~\ref{thm.multi}, is the stated system. If $b$ and $b'$ are solutions, then $U_t(b-b')=U_t(b)U_t(b')^{-1}$ acts trivially on $Q_1\otimes\cdots\otimes Q_t$, i.e., $b-b'\in Id_N(Q_1,\ldots,Q_t)$, and the converse also holds. Moreover, the number of nontrivial gates of $U_t(b)$ is $wt(b)$. For the second part, if $N=2$, then $2^s\mathcal{P}_s=\{0\}$ for $s\geq 1$ and $a_{J,b}=0$ for $|J|>t$, so the first part asks for $b\in\mathcal{P}_0^{\perp}$ with $b\cdot w_{J_0}=1$ and $b\cdot w_J=0$ for $J\neq J_0$, $|J|=t$, i.e., $b\in D^\perp\setminus(D+\langle w_{J_0}\rangle)^\perp$, which is nonempty if and only if $w_{J_0}\notin D$.
\end{proof}

\begin{example}\label{ex.address}\rm
We continue with Example~\ref{ex.cube}. Let $S\subseteq[m]$ with $1\leq|S|\leq\ell$, $x_S=\prod_{j\in S}x_j$, and $b=2^{\ell-|S|}\delta_{x_S}$. Since $h(x_S)=|S|$, Theorem~\ref{thm.general} gives $b\in H_N$. As $\delta_{x_S}\cdot ev(w)=1$ if $w=x_S$ and $0$ otherwise, Proposition~\ref{p:logact} gives $a_{J,b}=0$ for $J\neq S$ and $a_{S,b}=(-2)^{|S|-1}2^{\ell-|S|}=2^{\ell-1}\bmod N$. Hence, $U(b)$, which acts on the $2^{|S|}$ qubits indexed by the points $P_v$ with $v\mid x_S$, implements exactly the logical gate $C_S\text{-}Z$ (see~\cite[Logical Circuit~1]{Barg}). Thus, every multi-controlled $Z$ gate on at most $\ell$ logical qubits can be addressed with a diagonal transversal gate. For $m=\ell=3$, the logical $Z$, $CZ$, and $CCZ$ gates are obtained by applying $Z$, $S^{\pm1}$, and $T^{\pm 1}$, respectively, to $2$, $4$, and $8$ qubits, respectively, and, by Corollary~\ref{c:address} with $t=1$, these numbers are optimal, since they are the minimum weights of the cosets $b+Id_8$, where $Id_8=\mathrm{Span}_{\mathbb{Z}_8}\{2ev(1),4ev(x_1),4ev(x_2),4ev(x_3)\}$ (see Theorem~\ref{thm.tdmon}).

On the other hand, consider two copies of $Q$ with the same encoding and take $N=2$. For $b\in H_2^{(2)}(Q)$, Proposition~\ref{p:logactmb} gives the logical action
\[
\overline{U_2(b)}
=\prod_{i,j=1}^m
\left(C_{\{(1,i),(2,j)\}}\text{-}Z\right)^{c_{ij}},
\qquad
c_{ij}=b\cdot(w_i\star w_j)\bmod 2.
\]
Since $w_i\star w_j=w_j\star w_i$, we have $c_{ij}=c_{ji}$. Thus, for $i\neq j$, the logical gates $C_{\{(1,i),(2,j)\}}\text{-}Z$ and $C_{\{(1,j),(2,i)\}}\text{-}Z$ necessarily occur together, so neither can be addressed individually. For matching indices, $w_i\star w_i=w_i$. Moreover, $ev(1)\in\beta_2$ implies $w_i=ev(1)\star w_i\in\mathcal{P}_0$, and code preservation requires $b\in\mathcal{P}_0^\perp$. Hence,
\[
c_{ii}=b\cdot(w_i\star w_i)=b\cdot w_i=0\bmod 2.
\]
Consequently, gates between matching logical indices cannot occur, and no single logical $CZ$ between the two blocks can be addressed.

The same obstruction holds for $t\geq 2$ copies of $Q$ and $N=2$. If the indices of $J_0=\{(1,j_1),\ldots,(t,j_t)\}$ are not all equal, exchanging two blocks with different indices produces $J\neq J_0$ with $w_J=w_{J_0}$, so the two logical coefficients coincide. If all indices equal $j$, then
\[
w_{J_0}=w_j
=ev(1)\star w_j\star\cdots\star w_j\in\mathcal{P}_0,
\]
so the corresponding coefficient vanishes. Thus, no single logical $C^{t-1}Z$ gate between the blocks can be addressed.

More generally, for identically encoded copies of any CSS code and arbitrary $N$, the coefficients $a_{J,b}$ are invariant under permutations of the blocks. Therefore, every logical operator implemented by $U_t(b)$ must share this symmetry.
\end{example}

\section{Conclusion} 
\label{sec:conclusion}

In this paper, we considered the problem of computing the set of diagonal transversal gates that fix a CSS code. We determined the logical actions of these gates and the groups of transversal logical gates and of logical identities. We determined explicitly the set of equations defining the group. We computed the full group of diagonal transversal gates fixing the code, and of transversal logical gates, for any CSS code arising from monomial codes, thereby recovering and extending results in the literature on CSS-T codes, triorthogonal codes, and divisible codes. We also extended these results to diagonal gates acting on several code blocks, which determine the addressable transversal gates. In future work, we will consider extensions beyond monomial codes, such as algebraic geometry codes, and the implications of these results for magic state distillation or other methods to achieve universal fault-tolerant quantum computing.

\section{Acknowledgements}
 The National Science Foundation partially supported the first (DMS-2201075 and DMS-2401558), second  (DMS-2401558 and DMS-2502705), third (DMS-2201075 and DMS-2502705), fourth (CCF-2106189), and fifth (DMS-2401558) authors. The first three authors and the last author were partially supported by the Commonwealth Cyber Initiative. The fifth author was partially supported by an AMS-Simons Travel Grant, and Grant PID2022-138906NB-C21 funded by MICIU/AEI/10.13039/501100011033 and by ERDF/EU. Initial discussions for this collaboration took place while the third and fourth authors were visiting the Simons Institute for the Theory of Computing in Spring 2024.


\begin{thebibliography}{10}
\providecommand{\url}[1]{#1}
\csname url@samestyle\endcsname
\providecommand{\newblock}{\relax}
\providecommand{\bibinfo}[2]{#2}
\providecommand{\BIBentrySTDinterwordspacing}{\spaceskip=0pt\relax}
\providecommand{\BIBentryALTinterwordstretchfactor}{4}
\providecommand{\BIBentryALTinterwordspacing}{\spaceskip=\fontdimen2\font plus
\BIBentryALTinterwordstretchfactor\fontdimen3\font minus \fontdimen4\font\relax}
\providecommand{\BIBforeignlanguage}[2]{{%
\expandafter\ifx\csname l@#1\endcsname\relax
\typeout{** WARNING: IEEEtran.bst: No hyphenation pattern has been}%
\typeout{** loaded for the language `#1'. Using the pattern for}%
\typeout{** the default language instead.}%
\else
\language=\csname l@#1\endcsname
\fi
#2}}
\providecommand{\BIBdecl}{\relax}
\BIBdecl

\bibitem{shorPrimeFactorization}
\BIBentryALTinterwordspacing
P.~W. Shor, ``Polynomial-time algorithms for prime factorization and discrete logarithms on a quantum computer,'' \emph{SIAM J. Comput.}, vol.~26, no.~5, pp. 1484--1509, 1997. [Online]. Available: \url{https://doi.org/10.1137/S0097539795293172}
\BIBentrySTDinterwordspacing

\bibitem{shorFirstQECC}
\BIBentryALTinterwordspacing
------, ``Scheme for reducing decoherence in quantum computer memory,'' \emph{Phys. Rev. A}, vol.~52, pp. R2493--R2496, Oct 1995. [Online]. Available: \url{https://link.aps.org/doi/10.1103/PhysRevA.52.R2493}
\BIBentrySTDinterwordspacing

\bibitem{calderbankgoodquantum}
\BIBentryALTinterwordspacing
A.~R. Calderbank and P.~W. Shor, ``Good quantum error-correcting codes exist,'' \emph{Phys. Rev. A}, vol.~54, pp. 1098--1105, Aug 1996. [Online]. Available: \url{https://link.aps.org/doi/10.1103/PhysRevA.54.1098}
\BIBentrySTDinterwordspacing

\bibitem{cssoriginal2}
A.~{Steane}, ``{Multiple-Particle Interference and Quantum Error Correction},'' \emph{Proceedings of the Royal Society of London Series A}, vol. 452, no. 1954, pp. 2551--2577, Nov. 1996.

\bibitem{eastinknill}
\BIBentryALTinterwordspacing
B.~Eastin and E.~Knill, ``Restrictions on transversal encoded quantum gate sets,'' \emph{Phys. Rev. Lett.}, vol. 102, p. 110502, Mar 2009. [Online]. Available: \url{https://link.aps.org/doi/10.1103/PhysRevLett.102.110502}
\BIBentrySTDinterwordspacing

\bibitem{Bravi-12}
S.~Bravyi and J.~Haah, ``Magic-state distillation with low overhead,'' \emph{Physical Review A—Atomic, Molecular, and Optical Physics}, vol.~86, no.~5, p. 052329, 2012.

\bibitem{haah_classification_triorthogonal}
\BIBentryALTinterwordspacing
S.~Nezami and J.~Haah, ``Classification of small triorthogonal codes,'' \emph{Phys. Rev. A}, vol. 106, no.~1, pp. Paper No. 012\,437, 13, 2022. [Online]. Available: \url{https://doi.org/10.1103/physreva.106.012437}
\BIBentrySTDinterwordspacing

\bibitem{camps_triorthogonal}
E.~Camps-Moreno, H.~H. López, G.~L. Matthews, D.~Ruano, R.~San–José, and I.~Soprunov, ``Binary triorthogonal and {CSS-T} codes for quantum error correction,'' in \emph{2024 60th Annual Allerton Conference on Communication, Control, and Computing}, 2024, pp. 01--06.

\bibitem{kim_triorthogonal_selfdual}
\BIBentryALTinterwordspacing
M.~Shi, H.~Lu, J.-L. Kim, and P.~Sol\'e, ``Triorthogonal codes and self-dual codes,'' \emph{Quantum Inf. Process.}, vol.~23, no.~7, pp. Paper No. 280, 24, 2024. [Online]. Available: \url{https://doi.org/10.1007/s11128-024-04485-9}
\BIBentrySTDinterwordspacing

\bibitem{albert_transversal_gates}
S.~P. Jain and V.~V. Albert, ``Transversal clifford and {T}-gate codes of short length and high distance,'' \emph{IEEE Journal on Selected Areas in Information Theory}, vol.~6, pp. 127--137, 2025.

\bibitem{sanjose_good_triorthogonal}
R.~San-José, ``Asymptotically good binary triorthogonal codes and higher-level transversal gates,'' \emph{ArXiv 2609.08203}, 2026.

\bibitem{Rengaswamy-20}
N.~Rengaswamy, R.~Calderbank, M.~Newman, and H.~D. Pfister, ``On optimality of {CSS} codes for transversal {T},'' \emph{IEEE Journal on Selected Areas in Information Theory}, vol.~1, no.~2, pp. 499--514, 2020.

\bibitem{campsCSST}
E.~Camps-Moreno, H.~H. L\'opez, G.~L. Matthews, D.~Ruano, R.~San-Jos\'e, and I.~Soprunov, ``An algebraic characterization of binary {CSS-T} codes and cyclic {CSS-T} codes for quantum fault tolerance.'' \emph{Quantum Inf. Process.}, vol.~23, no. 230, 2024.

\bibitem{berardini_asymptotically_good_csst}
E.~Berardini, R.~Dastbasteh, J.~Etxezarreta~Martinez, S.~Jain, and O.~Sanz~Larrarte, ``Asymptotically good {CSS-T} codes and a new construction of triorthogonal codes,'' \emph{IEEE Journal on Selected Areas in Information Theory}, vol.~6, pp. 189--198, 2025.

\bibitem{felice_csst}
J.~Bolkema, E.~Andrade, T.~Dexter, H.~Eggers, V.~L. Fisher, L.~Szramowsky, and F.~Manganiello, ``{CSS-T} codes from {R}eed-{M}uller codes,'' \emph{IEEE Journal on Selected Areas in Information Theory}, vol.~6, pp. 199--204, 2025.

\bibitem{golowich_asymptotically_good_quantum_transversal}
\BIBentryALTinterwordspacing
L.~Golowich and V.~Guruswami, ``Asymptotically good quantum codes with transversal non-{C}lifford gates,'' in \emph{S{TOC}'25---{P}roceedings of the 57th {A}nnual {ACM} {S}ymposium on {T}heory of {C}omputing}.\hskip 1em plus 0.5em minus 0.4em\relax ACM, New York, [2025] \copyright 2025, pp. 707--717. [Online]. Available: \url{https://doi.org/10.1145/3717823.3718234}
\BIBentrySTDinterwordspacing

\bibitem{prakash_qutrit_magic_state}
\BIBentryALTinterwordspacing
S.~Prakash and T.~Saha, ``Low {O}verhead {Q}utrit {M}agic {S}tate {D}istillation,'' \emph{{Quantum}}, vol.~9, p. 1768, Jun. 2025. [Online]. Available: \url{https://doi.org/10.22331/q-2025-06-12-1768}
\BIBentrySTDinterwordspacing

\bibitem{wills_constant_overhead_msd}
A.~Wills, M.-H. Hsieh, and H.~Yamasaki, ``Constant-overhead magic state distillation,'' \emph{Nature Physics}, vol.~21, pp. 1842--1846, 2025.

\bibitem{Webster-23}
M.~A. Webster, A.~O. Quintavalle, and S.~D. Bartlett, ``Transversal diagonal logical operators for stabiliser codes,'' \emph{New Journal of Physics}, vol.~25, no.~10, p. 103018, 2023.

\bibitem{Haah-18}
J.~Haah, ``Towers of generalized divisible quantum codes,'' \emph{Physical Review A}, vol.~97, no.~4, p. 042327, 2018.

\bibitem{krishna2018polar}
A.~Krishna and J.-P. Tillich, ``Magic state distillation with punctured polar codes,'' \emph{ArXiv 1811.03112}, 2018.

\bibitem{Barg}
\BIBentryALTinterwordspacing
A.~Barg, N.~J. Coble, D.~Hangleiter, and C.~Kang, ``Geometric structure and transversal logic of quantum {R}eed-{M}uller codes,'' \emph{IEEE Trans. Inform. Theory}, vol.~72, no.~1, pp. 415--436, 2026. [Online]. Available: \url{https://doi-org.ezproxy.lib.vt.edu/10.1109/tit.2025.3592631}
\BIBentrySTDinterwordspacing

\bibitem{coderings}
S.~T. Dougherty, \emph{Algebraic coding theory over finite commutative rings}.\hskip 1em plus 0.5em minus 0.4em\relax Springer, 2017.

\bibitem{Nielsen_Chuang_2010}
M.~A. Nielsen and I.~L. Chuang, \emph{Quantum Computation and Quantum Information: 10th Anniversary Edition}.\hskip 1em plus 0.5em minus 0.4em\relax Cambridge University Press, 2010.

\bibitem{ketkar_nonbinary_stabilizer}
A.~Ketkar, A.~Klappenecker, S.~Kumar, and P.~K. Sarvepalli, ``Nonbinary stabilizer codes over finite fields,'' \emph{IEEE Trans. Inform. Theory}, vol.~52, no.~11, pp. 4892--4914, 2006.

\bibitem{Anderson-14}
J.~T. Anderson and T.~Jochym-O'Connor, ``Classification of transversal gates in qubit stabilizer codes,'' \emph{Quantum Info. Comput.}, vol.~16, no. 9–10, p. 771–802, Jul. 2016.

\bibitem{calderbankclassicalcsst}
N.~Rengaswamy, R.~Calderbank, M.~Newman, and H.~D. Pfister, ``Classical coding problem from transversal {T} gates,'' in \emph{2020 IEEE International Symposium on Information Theory (ISIT)}, 2020, pp. 1891--1896.

\bibitem{dualev}
H.~H. L{\'o}pez, I.~Soprunov, and R.~H. Villarreal, ``The dual of an evaluation code,'' \emph{Designs, Codes and Cryptography}, vol.~89, no.~7, pp. 1367--1403, 2021.

\bibitem{Hu-22}
J.~Hu, Q.~Liang, and R.~Calderbank, ``Divisible codes for quantum computation,'' \emph{arXiv preprint arXiv:2204.13176}, 2022.

\bibitem{nguyen_good_css_ccz}
\BIBentryALTinterwordspacing
Q.~T. Nguyen, ``Good binary quantum codes with transversal {CCZ} gate,'' in \emph{S{TOC}'25---{P}roceedings of the 57th {A}nnual {ACM} {S}ymposium on {T}heory of {C}omputing}.\hskip 1em plus 0.5em minus 0.4em\relax ACM, New York, 2025, pp. 697--706. [Online]. Available: \url{https://doi-org.ezproxy.lib.vt.edu/10.1145/3717823.3718186}
\BIBentrySTDinterwordspacing

\bibitem{he_addressable_transversal}
Z.~He, V.~Vaikuntanathan, A.~Wills, and R.~Y. Zhang, ``Quantum codes with addressable and transversal non-{C}lifford gates,'' \emph{arXiv preprint arXiv:2502.01864}, 2025.

\bibitem{cui_diagonal_gates}
S.~X. Cui, D.~Gottesman, and A.~Krishna, ``Diagonal gates in the {C}lifford hierarchy,'' \emph{Phys. Rev. A}, vol.~95, p. 012329, 2017.

\bibitem{he_good_addressable}
Z.~He, V.~Vaikuntanathan, A.~Wills, and R.~Y. Zhang, ``Asymptotically good quantum codes with addressable and transversal non-{C}lifford gates,'' \emph{arXiv preprint arXiv:2507.05392}, 2025.

\end{thebibliography}


\end{document}